\documentclass[12pt,oneside,english]{article}
\usepackage{amsthm}
\usepackage{amsmath}
\usepackage{amsfonts}
\usepackage[utf8]{inputenc}
\usepackage{amssymb}

\usepackage{float,todonotes}

\usepackage[a4paper]{geometry}
\usepackage{verbatim}
\usepackage{pdfpages}
\usepackage{nomencl}
\usepackage[resetfonts]{cmap} %helps dealing with ligatures (special symbols for ff, fi etc that prevent searching)

\providecommand{\makenomenclature}{\makeglossary}
\makenomenclature
\usepackage{thmtools}
\usepackage{thm-restate}

\usepackage{tikz}
\usetikzlibrary{automata,positioning,chains}
\usetikzlibrary{arrows}
\usepackage{varwidth}
\usepackage{ragged2e}
\usepackage{mathtools}

\makeatletter

\usepackage[linesnumbered,ruled,vlined]{algorithm2e}

\SetAlFnt{\small}
\SetAlCapFnt{\small}
\SetAlCapNameFnt{\small}
\SetAlCapHSkip{0pt}
\IncMargin{-\parindent}

\SetKwRepeat{Do}{do}{until}

\usepackage[numbers]{natbib}

\let\originalleft\left
\let\originalright\right
\renewcommand{\left}{\mathopen{}\mathclose\bgroup\originalleft}
\renewcommand{\right}{\aftergroup\egroup\originalright}

\newcommand{\ceil}[1] {\left\lceil{#1}\right\rceil}
\newcommand{\floor}[1] {\left\lfloor{#1}\right\rfloor}

\usepackage{xparse}

\newcommand{\Exp}[1]{\mathrm{E}\left[#1\right]}
\DeclareDocumentCommand{\Exp}{ m g }{
    {%
    \mathrm{E}
        \IfNoValueF {#2} {_{#2}}
    \left[#1\right]
    }%
}

\theoremstyle{plain}
\newtheorem{theorem}{Theorem}[section]
\newtheorem{corollary}[theorem]{Corollary}
\newtheorem{lemma}[theorem]{Lemma}

\theoremstyle{definition}
\newtheorem{definition}[theorem]{Definition}

\theoremstyle{remark}

\newcommand{\ds}{\mathrm{ds}}
\newcommand{\dsv}{\mathrm{dsv}}

\newcommand{\auctionutility}{\eta}
\newcommand{\auctionvalue}{\zeta}

\newcommand{\ermv}{\mathrm{em}}
\newcommand{\sw}{\mathrm{sw}}
\newcommand{\rand}{\mathrm{\tt rand}}
\newcommand{\comp}{\mathrm{\tt composite}}
\newcommand{\stay}{\mathrm{\tt stay}}
\newcommand{\match}{\mathrm{\tt match}}
\newcommand{\opt}{\mathrm{opt}}
\newcommand{\rd}{\rho}
\newcommand{\expec}{\mathop{{}\mathbb{E}}}

\newcommand{\Reals}{\mathbb{R}}

\newcommand{\pass}{b}
\newcommand{\taxi}{t}
\newcommand{\buyer}{\mbox{\rm buyer}}

\newcommand{\val}{\mbox{\rm value}}
\newcommand{\loc}{\operatorname{loc}}
\newcommand{\price}{p}
\newcommand{\dist}{\mbox{\rm dist}}

\usepackage[font={small,bf}, labelfont={small,bf}, margin=1cm]{caption}
\title{\textbf{Flow Equilibria via Online Surge Pricing}}
\author{Amos Fiat, Yishay Mansour, Lior Shultz}
\author{Amos Fiat  \\
	Tel Aviv University   \\
	\and
	Yishay Mansour \\
	Tel Aviv University \\
	Google \\
	\and
	Lior Shultz \\
	Tel Aviv University \\
}

\date{\today}

\usepackage{babel}
\begin{document}

%\prelimpages
%\includepdf[pages=-]{English_cover}

%\epigraph{``Relative to the taxi industry, Uber is a sustaining innovation; that is, it makes customers' lives better."}{--- \textup{Clayton M. Christensen}, Harvard University\\ Coined ``Disruptive Innovation"}

\newcommand{\subinput}{\subsection}
\newcommand{\subsubinput}{\subsubsection}
\newcommand{\cursupply}{s^{t-1}}
\newcommand{\newsupply}{s^t}
\newcommand{\newdemand}{d^t}
\newcommand{\surgevec}{r^t}
\newcommand{\curflow}{f^t}
\newcommand{\flowutility}{\mu^t}

\newcommand{\paper}{paper}
\newcommand{\Paper}{Paper}

\newif\ifonline
\onlinetrue % comment out to hide answers

\maketitle

\section*{Abstract}
 
We explore issues of dynamic supply and demand in ride sharing services such as Lyft and Uber, where demand fluctuates over time and geographic location. We seek to maximize social welfare which depends on taxicab locations,  passenger locations, passenger valuations for service, and the distances between taxicabs and passengers. Our only means of control is to set surge prices, then taxicabs and passengers
maximize their utilities subject to these prices.

We study two related models: a continuous passenger-taxicab setting, similar to the Wardrop model, and a discrete (atomic) passenger-taxicab setting. In the continuous setting, every location is occupied by a set of
infinitesimal strategic taxicabs and a set of infinitesimal non-strategic passengers. In the discrete setting every location is occupied by a set of strategic agents, taxicabs and passengers, passengers have differing  values for service.

We expand the continuous model to a time-dependent setting and study the corresponding online environment.

The utility for a strategic taxicab that drives from $u$ to $v$ and picks up a passenger at $v$ is the surge price at $v$ minus the distance from $u$ to $v$. The utility for a strategic passenger at $v$ that gets service is the value of the service to the passenger minus the surge price at $v$.

Surge prices are in passenger-taxicab equilibrium if there exists a min cost flow that moves taxicabs about such that (a) every taxicab follows a best response, (b) all strategic passengers at $v$ with value above the surge price $r_v$ for $v$, are served and (c) no strategic passengers with value below $r_v$ are served (non-strategic infinitesimal passengers are always served).

This \paper{} computes surge prices such that resulting passenger-taxicab equilibrium maximizes social welfare, and the computation of such surge prices is in poly time. Moreover, it is a dominant strategy for passengers to reveal their true values.

We seek to maximize social welfare in the online environment, and derive tight competitive ratio bounds to this end.
Our online algorithms make use of the surge prices computed over time and geographic location, inducing successive passenger-taxicab equilibria.

%\tableofcontents*

% acknowledgments{
	
%I would like to thank my supervisors Prof. Amos Fiat and Prof. Yishay Mansour for guiding me in the rewarding process of research.

%Due to my limiting circumstances they have forfeited much of their personal time to meet and discuss the various work presented in this thesis, for that I am truly thankful. They have allowed me to experience many different methods of work and approaches and I cannot imagine this process without either of them. Thank you very much!

%}

%\textpages
%\listoffigures
%\printnomenclature{}

\section{Introduction}

In the sharing economy\footnote{Also known as the ``gig" economy.}
individual self-interested suppliers compete for customers.
According to PWC, the {\sl sharing economy} is projected to exceed
$300$  billion USD within 8 years. Lyft and Uber are prime examples of such systems. According to  \cite{lam2017demand,NBERw22627} it is the users who gain the majority of the surplus from such systems, and significantly so. Contrawise, many studies suggest negative societal issues in the sharing economy  (e.g., see
 \cite{MARTIN:2016,Cramer:2016,RICHARDSON:2015,berger2017drivers}).

Unlike salaried employees of livery firms, drivers for Uber (and
other ``gig" suppliers) are free to decide when they are working and
what calls/employment to accept. {\sl E.g.}, drivers can refuse to
accept a call if it is too far away. To increase supply (and reduce
demand) Uber introduced ``surge pricing" which is a multiplier on
the base price when demand outstrips supply. The surge price can be
different at different locations.

In the past pricing schemes resulted in what was theorized to be negative work elasticity \cite{camerer1997labor}.
In their work it is suggested that drivers impose upon themselves ``income targets". This means
that drivers will work until they reach their target income for the day causing them to extend
their hours in times of low payouts.
Recent studies suggest that this is false, surging prices in times of peak demand seems to conjure positive work elasticity \cite{Chen:2016:DPL:2940716.2940798}, allowing supply and demand to balance more efficiently.

\subinput{Network Model, Surge Pricing, Utility, and Passenger-Taxicab Equilibria}

Our goal is to maximize social welfare, defined as the sum of valuations of the users serviced by taxicabs, minus the cost associated with providing such service. We do so by setting surge prices (one per location), and let the system reach equilibrium. Our surge pricing schemes have several additional features such as envy freeness.

We consider two related settings:

\begin{itemize} \item A continuous setting where supply and demand consist of infinitesimal quanta,  {\em supply} and {\em demand} are modeled as fractional
quantities at locations. This is analogous to the non-atomic traffic model used in
Wardrop equilibria \cite{Wardrop:1952}. \begin{itemize} \item Here we assume that the taxicabs  are strategic and respond to changing surge prices whereas passengers are non-strategic so that demand is insensitive to price (alternately, one may view these passengers as having high value for service). \item The cost for a taxicab at location $x$ to serve a customer at location $y$ is the distance from $x$ to $y$.
\item Our goal here is to set a surge price $r_x$ at every location $x$ so as to incentivize taxicabs to act in a way that maximizes social welfare, {\sl i.e.}, all possible demand is serviced while the sum of distances traversed is minimized.
 \end{itemize}
\item A discrete setting where both taxicabs and passengers are strategic, and every taxicab and passenger is associated with some location. \begin{itemize} \item In this setting both demand and supply may change as a function of the surge price. Every passenger has a value for service and every taxicab has a cost for service at a given location, {\sl e.g.}, the distance to the location. \item Our goal here is to maximize social welfare (the sum of the values for the served customers minus the sum of costs of the taxicabs to do so). \item At every location $x$, we set a surge price $r_x$, that incentivizes taxicabs to serve passengers in a manner that maximizes social welfare.
    \item Moreover, maximizing social welfare is not only in equilibrium but also envy free.
    \item Every passenger at $x$ whose value is strictly greater than $r_x$ is served, and no passenger at $x$ with value strictly less than $r_x$ is served. \end{itemize}
\end{itemize}

%Both demand and supply are located on the vertices of a metric graph.

We define the utility for a taxicab at $x$ to serve a passenger at $y$ as the surge price at $y$, $r_y$, minus the distance from $x$ to $y$. A passenger at $x$ with value $v$ has utility $v-r_x$ to be served by a taxicab, and utility zero if she takes no taxicab. Clearly, a passenger at $x$ with $v-r_x<0$ will refuse to take a taxicab.

We introduce the notion of a {\sl passenger-taxicab equilibria}, for both continuous and discrete settings. A flow is a mapping from the current supply to some new supply. A flow has an associated cost which is the sum over edges of the flow along the edge times the length of the edge.
A flow $f$ is said to be a min cost flow that maps the current supply to the new supply if it achieves the minimal cost for moving the current supply to the new supply (this cost is also called the min earthmover cost).

A passenger-taxicab equilibria consists of a vector of surge prices $r=\langle r_x\rangle$, where $r_x$ is the surge price at location $x$, current supply $s=\langle s_x\rangle$, new supply $s'=\langle s'_x \rangle$ and demand $d=\langle d_x \rangle$, such that, for any min cost flow from $s$ to $s'$, every taxicab and every passenger maximize their utility. {\sl I.e.}, no taxicab can improve its utility by doing anything other than following the flow, every passenger at $x$ who has value greater than $r_x$ is in $d_x$ and is served. Every passenger at $x$ who has value less than $r_x$ is not served. 

The surge prices $r_x$ are poly time computable. In the continuous setting this is polynomial in the number of locations, in the discrete setting this  is polynomial in the number of passengers and taxicabs.

\subinput{Maximizing Social Welfare in an Online Setting via Surge Pricing}

We consider an online setting based on the continuous setting, where the time progresses in discrete
time steps. In each time step the following occurs: First, a new
demand allocation appears. Second, the online algorithm determines a
new supply. Given an allocation of supply and demand, the demand
served at a location is the minimum between the supply and demand at
the location. The social welfare is the difference between the total
demand served and the total movement cost, summed over all locations
and time steps. The main new crux of our model is that the online
algorithm (principal) can not impose a new supply allocation, but is
limited to setting surge prices. If flow $f$ is a flow equilibrium
arising from these surge prices --- strategic suppliers follow flow
$f$. Our results on surge prices for flow equilibria imply that the
online algorithm has flexibility in selecting the desired supply.

Trivially, for any metric, a simple algorithm that randomizes the start setting and doesn't move achieves a $\Theta(1/k)$ competitive ratio, where $k$ is the number of locations. However, If the costs of moving from any location to any other location is 1, we give an optimal competitive ratio of $\Theta(\sqrt{1/k})$.
If the demand sequence has the property that at any time and location the demand does not exceed $1/\rho$ ($\rho \geq 1$),
then we show a tight competitive ratio bound of  $\Theta(\sqrt{\rho/k})$.
For more general metric spaces we show mainly negative results.
Specifically, if all the distances are $1+\epsilon$ we show that the competitive ratio is no better than  $(1+\epsilon)^2/(\epsilon k)$, which implies
an optimal competitive ratio of $\Theta(1/k)$ for $\epsilon =\Theta(1)$.

Another extension we consider is when the average difference between
successive demand vectors is bounded by $\delta$ (in total variation
distance). In this case we show that simply matching supply to the
current demand gives a competitive ratio of $1-\delta$ and show that
the competitive ratio can not be better than $1-\delta/4$ (in the
case that all the distances are $1$).

\subinput{Related Work}
% Summary - \cite{banerjee2015pricing} Here they discuss a "queueing-theoretic" model and attempt to study dynamic pricing
% The analysis tries to simulate earning by attempting to bequeath a constant percentage of the earning to the platform mechanism with the driver retaining the remainder
%

It has been observed in taxicab services that a mismatch between
supply and demand, along with first-in-first-out scheduling of
service calls, without restricting the ``call radius", results in
reduced efficiency and even market failure
\cite{ARNOTT:1996,YANG:2002}. This happens because taxicabs are
dispatched to pick up customers at great distance
%from the customer
because no closer taxicab is currently available, more time is
wasted traveling to pick up clients, and the system performance
degrades.
Recent papers \cite{Chen:2016,Castillo:2017} study how changing
surge prices over time allow one to avoid such issues. These papers
do not consider the issue of having geographically varying surge
prices.

Assuming a stochastic passenger arrival rate, \cite{banerjee2015pricing} uses a queue theoretic approach
to model driver incentives in the system. The paper considers a simplistic dynamic pricing scheme, where there are two different pricing schemes for each node depending on the amount of drivers at said node. This model is compared to a simple flat rate. Drivers are assumed to calculate their incentives over several rides. The paper concludes that the dynamic pricing scheme can only achieve the welfare of the flat rate. However, the dynamic pricing scheme allows for the manager to have more room for error in calculating what the optimal rates are.

A central problem in handling a centralized taxi system involves routing empty cars between regions . Within the centralized mechanism, \cite{braverman2016empty} shows that, assuming stochastic arrival of passengers, an optimal static strategy ({\sl i.e.}, one that does not change it's routing policy based on current shortages) can be calculated by solving a linear programming problem.

Recently, and independently, a similar problem was studied in \cite{ma2018spatio}. In their model, selfish taxicabs seek to maximize revenue over time. There is no explicit cost for travel, one loses opportunities by taking long drives. They derive prices in equilibria that maximize the sum of passenger valuations, but ignore travel costs. In contrast, we ignore the time dimension and focus on the passenger valuations and travel costs.

Competitive analysis of online algorithms
\cite{ST85a,ST85b,Karlin1988} considers a worst case sequence of
online events with respect to the ratio between the performance of
an online algorithm and the optimal performance.
In a centralized setting, task systems, \cite{Borodin92}, can be
used to model a wide variety of online problems. Events are
arbitrary vectors of costs associated with different states of the
system, and an online algorithm may decide to switch states (at some
additional cost). A strategic version of this problem, for a single
agent, was considered in \cite{DBLP:conf/soda/CohenEFJ15} where a
deterministic incentive compatible mechanism was given. The
competitive ratio for incentive compatible task system mechanisms is
$O(1/k)$ where $k$ is the number of states. We cannot use the
incentive compatible task system mechanisms from
\cite{DBLP:conf/soda/CohenEFJ15} for two reasons: (1) in our setting
there are a large number of strategic agents (many Uber drivers)
split amongst a variety of different [task system] states
(locations) rather than one such agent in a single state, and (2)
the suppliers have both profits (payments) and loss (relocation).

Competitive analysis of the famous $k$-server problem
\cite{MANASSE1990208} has largely driven the field of online
algorithms. A variant of the $k$-server problem is known as the
$k$-taxicab problem \cite{Fiat:90,XIN:2004}. Although the problem we
consider herein and the $k$-taxicab problem both seek efficient
online algorithms, and despite the name, the nature of the
$k$-taxicab problem is quite different from the problem considered
in this \paper{}. In the $k$-taxicab problem a single request occurs at
discrete time steps and a centralized control routes taxicabs to
pick up passengers, seeking to minimize the distances traversed by
taxis while empty of passengers. Taxicabs are not selfish suppliers,
and all requests must be satisfied. This is quite different from our
setting where both demand and supply are spread about
geographically, there are many strategic suppliers, and not all
demand must be served.

\section{Model and Notation}\label{sec:model}

\subinput{The Continuous Passenger-Taxicab Setting}

We model the network as a finite metric space $G=(V,E)$, where
$\ell_{u,v}\geq 0$ is the distance between
vertices $u, v \in V$. {\sl I.e.}, $\ell_{u,v}$ is the cost to a
taxicab to switch between vertices $u$ and $v$. Infinitesimally small taxicabs reside in
the vertices $V$.

Demand and supply are vectors in $[0,1]^{|V|}$ that sum to one.
Given demand $d$ and current supply $s$,
we incentivize strategic taxicabs so that current supply $s$ becomes new supply $s'$ which services the demand $d$.

If the demand in vertex $u$ is $d_u$, and the new supply in
vertex $u$  is $s'_u$, then the minimum of the two is the
actual demand served (in vertex $u$). Note that if the two
are not identical then there are either unhappy passengers (without
service) or unhappy taxicabs (with no passengers to service).
 Formally,
\begin{definition}\label{def:demandserved}
we define the {\sl demand served}, as follows:
\begin{itemize} \item The {\sl demand served} in vertex $u$, $\ds(s'_u,d_u)$, is the minimum of $s'_u$ and $d_u$, {\sl i.e.,} $\ds(s'_u,d_u)=\min(s'_u,d_u)$.
\item Given a demand vector $d$ and a supply vector $s'$, the total demand served is
$\ds(s',d)= \sum_{u\in V} \ds(s'_u,d_u)=
\sum_{u\in V}{\min(s'_u,d_u)}$.
\end{itemize}
\end{definition}

Switching supply from $s$ to $s'$ is implemented via a flow $f$. A
flow from $s$ to $s'$ is a function $f(u,v):V\times V \mapsto
\Reals^{\geq 0}$ that has the following properties:
\begin{itemize}
    \item For all $u,v\in V$, $f(u,v)\geq 0$.
    \item For all $v\in V $, $\sum_{u\in V} f(u,v)=s'_v $.
    \item For all $u\in V$, $\sum_{v\in V} f(u,v)=s_u $.
\end{itemize}

%Given a flow $f$ from supply vector $s$ to supply vector $s'$ we say that the flow $f$ {\sl induces} supply $s'$.
We define the earthmover distance between supply vectors,
\begin{definition}\label{def:earthmover}
The cost of flow $f$ is
$\ermv(f)=\sum_{u,v,\in V} f(u,v)\ell_{u,v}$.
The earthmover distance from supply vector $s$ to supply vector $s'$ is
$$\ermv(s,s')=\min_{\mbox{\rm flows $f$ from $s$ to $s'$}} \ermv(f).$$
\end{definition}

We assume that switching supply from $s$ to $s'$ is implemented via a flow $f$ of minimal cost.
Note that there may be multiple flows with the same minimal cost  --- see Figures \ref{fig:flow1} and \ref{fig:flow2}.

%\subinput{Surge Pricing and Flow Equlibria}

In order to incentivize our strategic taxicabs to move to a new supply vector, we use surge
pricing in vertices.

\begin{definition}\label{def:surgepricing} Surge pricing is a vector, $r\in \Reals^{\geq 0}$, where $r_v$ is the
payment to a taxicab that serves demand in vertex $v\in V$.
\end{definition}

We define the utility for an infinitesimal taxicab, given surge
pricing $r$, as follows.

\begin{definition}\label{def:utility} Given supply $s$, new supply $s'$, surge prices $r$, demand $d$, and a min cost flow $f$ from $s$ to $s'$, the utility for
a taxicab that switches from vertex $u$ to vertex $v$   is
$$\mu(u\mapsto v|s',r,d) = r_v \cdot
\left(\frac{\ds(s'_v,d_v)}{s'_v}\right) -
  \ell_{u,v}.$$
  \end{definition}

To motivate the above definition of utility  $\mu(u\mapsto v|s',r,d)$, of
switching from $u$ to $v$, consider the following:
\begin{itemize}
\item The probability of serving a passenger in vertex $v$ is $\frac{\ds(s'_v,d_v)}{s'_v}$. This follows since:
 \begin{itemize}
\item If passengers outnumber taxicabs in vertex $v$ then any such taxicab will surely serve a passenger.
\item Alternately, if taxicabs outnumber passengers in vertex $v$ then the choice of which taxicabs serve passengers is a random subset of the taxicabs. \end{itemize}
\item The profit from serving a passenger in vertex $v$ is equal to the surge price for that vertex, $r_v$.
\item The cost of serving a passenger in vertex $v$, given that the taxicab was previously in vertex $u$, is $\ell_{u,v}$.
    \end{itemize}

Finally, we define the notion of a passenger-taxicab equilibrium, where no
infinitesimal taxicab can benefit from deviations.

\begin{definition}\label{def:flowequilibrium}
Given a demand vector $d$, current supply vectors $s$, and new supply $s'$, we say that
a surge pricing $r$ is in {\em passenger-taxicab equilibrium}, if
for every min cost flow $f$ from $s$ to $s'$,
 for every $u,v\in V$ such that
$f(u,v)>0$  we have that \begin{equation}\mu(u\mapsto v|s',r,d) =\max_{w\in V}
\mu(u\mapsto w|s',r,d).\label{eq:incentivesequilibrium}\end{equation} {\em I.e.}, every infinitesimal taxicab is
choosing a best response.
Such a passenger-taxicab equilibrium is said to {\em induce supply
$s'$.}
\end{definition}

Our goal in the continuous setting is to set surge prices so that the new supply $s'=d$ is a passenger-taxicab equilibrium.

In this continuous setting we take demand $d$ to be insensitive to the surge prices. %The current supply $s$, the new supply $s'$, and the surge prices $r$
%are in passenger-taxicab  is the result
%of a flow equilibrium, which itself depends on the current supply $s$, the surge pricing $r$, and the [price-insensitive] demand $d$.
%
In the next section we describe the discrete setting where both the demand and the supply are sensitive to the prices.
One could define a  continuous passenger-taxicab setting where every location has an associated density function for passenger valuations. 
Then, we could convert this continuous setting to an instance of the discrete passenger-taxicab setting with $1/\epsilon$ taxicabs/passengers. Under appropriate conditions, this will give a good approximation to a continuous passenger-taxicab setting where both demand and supply are sensitive to surge pricing.

%Fix a supply vector, $s$, and a demand vector
%$s'$. Theorems \ref{thm:surge} and \ref{thm:anyst}, and
%associated algorithms, given a desired supply
% compute surge prices, $r$, which induce a
%flow equilibria that induces supply vector $s'$.

\subinput{The Discrete Passenger-Taxicab Setting}

%We consider an additional model, allowing passengers to submit a valuation for the ride.
%In this case, after submitting valuations for the rides surge prices are posted for each location and a passenger
%is said to be interested in a ride if the surge price at his location is lower or equal to his valuation.

%The social welfare in this case is defined as the sum of the valuations of the passengers served minus the cost
%of serving the load (the cost of movement for each taxi that serves any of the served passengers)

As above, we model the network as a finite metric space $G=(V,E)$, and the cost to a
taxicab to switch between vertices $u$ and $v$ is the
distance between them, $\ell_{u,v}$. Unlike the continuous case, there is an integral number of taxicabs and passengers at every vertex.

Let $B=\{\pass_1, \ldots, \pass_m\}$ be a set of $m$ passengers and
$T=\{\taxi_1, \ldots , \taxi_n\}$ be a set of $n$ taxicabs. Every passenger $\pass_i\in B$ has a value $\val(\pass_i)\geq 0$ for service.
A supply $s$ is a vector $s = \langle s_v
\rangle_{v\in V}$  where $s_v\subseteq T$ for all $v
\in V$, $\cup_{v\in V} s_v=T$, and $s_v\cap s_u = \emptyset$ for all $u, v \in V$, $u\neq v$.

A profile $P$ is a partition of the passengers $B$, where for each $u\in V$ the set $P_u \subseteq B$ is the set of passengers at $u$.
A demand is a function of a vertex and a surge price at the vertex. We define the function $d_v$ as follows: $$d_v(r_v)= \{ \pass_i \in P_v | \val(\pass_i)\geq r_v \}.$$ Ergo, $d_v(r_v)$ is the set of passengers at vertex $v$ that are interested in service given that the price is $r_v$, {\sl i.e.}, those passengers whose value is at least $r_v$. Note that $d_v(0)=P_v$.

For ease of notation, we denote a collection of entities $x_v$ for each vertex $v\in V$, by $x=\langle x_v\rangle_{v\in V}$.
For example, $s=\langle s_v\rangle_{v\in V}$, $d=\langle d_v \rangle$, and $r=\langle r_v\rangle_{v\in V}$.

%The location of a passenger $\pass_i$ and taxi $\taxi_j$ is denoted by $d(\pass_i)\in V$ and $s(\taxi_j)\in V$ respectively.

%

%
%We abuse the notation and identify passenger $\pass_i$ and taxicab $\taxi_j$ with `location' $i$ and $j$.
%We denote by
%^$\ell_{i,j}=\ell_{\loc(\pass_i),\loc(\taxi_j)}$.
%Also, for some location $\gamma$ denote $\ell_{\gamma,\taxi_j}=\ell_{\gamma,j}$.

%Given $B$ we define a demand $d$ where $d_v= |\{i : \loc(\pass_i)\}|$. Similarly,
%given $T$ we define a supply $s$ where $s_v= |\{i : \loc(\taxi_i)\}|$
%Given demand $D$ we define the aggregate demand at vertex $v$ as $\demandB_v=\{\pass_i:d(\pass_i)=v\}$.
%Similarly, given supply $s$ we define the aggregate supply at node $v$ as  $\supplyT_v=\{\taxi_j:s(\taxi_j)=v\}$.

Define a flow $f$ from supply $s$ to supply $s'$ as follows.
The flow $f(x,y):V\times V \mapsto \mathbb{Z}^+$ has the following properties:
\begin{itemize}
	\item For all $u,v\in V$, $f(u,v)\in \mathbb{Z}^+$.
	\item For all $u\in V$, $\sum_{v\in V} f(u,v)=|s_u|$.
	\item For all $v\in V$, $\sum_{u\in V} f(u,v)=|s'_v|$.
\end{itemize}
The flow from a vertex $u$ is equal to the number of taxicabs at $u$ under supply $s$, {\sl i.e.}, $|s_u|$. The flow into a vertex $v$ is equal to the number of taxicabs at $v$ under supply $s'$, {\sl i.e.}, $|s'_v|$.
The cost of a flow in the discrete setting is the same as the cost of a flow in the continuous setting (Definition \ref{def:earthmover}), {\sl i.e.}, $\sum_{u,v\in V} f(u,v)\ell_{u,v}$.

We now define the demand served at a vertex $u$,
\begin{definition}\label{def:demandserved}
For a vertex $v$, given a supply $s'_v$, a surge price $r_v$, and a demand $d_v(r_v)$, we define the {\sl demand served},
$\ds_v(s'_v,d_v,r_v)\subseteq P_v$, as follows:
\begin{itemize}
\item
If $|d_v(r_v)| \leq |s'_v|$ then $\ds_v(s'_v,d_v,r_v) = d_v(r_v)$.
\item
If $|s'_v| < |d_v(r_v)|$ then  $\ds_v(s'_v,d_v,r_v)$ is the set of the $|s'_v|$ highest valued passengers from $d_v(r_v)$, breaking ties arbitrarily.
\end{itemize}
%Let $s'=\langle s_v\rangle_{v\in V}$, $d=\langle d_v\rangle_{v\in V}$, and $\langle r_v\rangle_{v\in V}$.
Given demand functions $d$, surge prices $r$, and new supply $s'$, the total demand served $\ds(s',d,r)$ and its value $\dsv(s',d,r)$ is given by
\begin{eqnarray*} \ds(s',d,r)&=& \cup_{v\in V} \ds_v(s'_v ,d_v,r_v);\\
\dsv(s',d,r)&=&
\sum_{\pass_i\in\ds(s' ,d,r)} \val(\pass_i). \end{eqnarray*}

%\begin{itemize} \item The {\sl demand served} at vertex $v$:
%%We sort the passengers $B_v$ according to their values, and the
%$\ds_v(s'_v,d_v,r_V) = d_v(r_v)$ if $|d_v(r_v)| \leq |s'_v|$. If $|s'_v| < |d_v(r_v)|$ then set $\ds_v(s'_v,d_v,r_v)\subset d_v(r_v)$, where $|\ds_v(s'_v,d_v,r_v)|=|s'_v|$ and $\ds_v(s'_v,d_v,r_v)$ are the highest valued passengers from $d_v(r_v)$, breaking ties arbitrarily.
%%$\ds(s'_i,s'_i)$, is the minimum of $s'_i$ and $s'_i$.
%\item Given demands $d_v$, supplies $s'_v$, and surcharges $r_v$ the  demand served is
%$\ds(s',d,r)= \cup_{v\in V} \ds_v(s'_v ,d_v,r)$ and its value is
%$\dsv(f,r)=
%\sum_{\pass_i\in\ds(s' ,d,r)} \val(\pass_i)$.
%\end{itemize}
\end{definition}

%The value of a flow is the sum of the values of the passengers served. We assume that the passengers with
%the highest value in some vertex are the first to get served. {\sl I.e.}, if the are some $s'_v$ taxis at some vertex
%$v$ then the $s'_v$ passengers with the highest valuations will receive service.

\begin{definition}\label{def:sw}
The social welfare is the difference between the sum of the values of the passengers served and the cost of the min cost flow,
which is the sum of the distances traveled by the taxis. Namely, for current supply $s$, new supply $s'$, demand functions $d$, and surge prices $r$, the social welfare is
%{\sl I.e. }, $\sum_{u,v\in V} f(u,v)\cdot\ell_{u,v}$.
\begin{equation}
SW(s,s',r,d)=\dsv(s',d,r)-em(s,s'). \label{eq:swdiscrete}
\end{equation}
\end{definition}

Remark: we did not define social welfare in the continuous passenger-taxicab setting where the passengers are price insensitive.
However, one can view the social welfare in the price-insensitive demand setting as a special case of the responsive demand setting when all passenger valuations are very high.

%We define the social welfare of a flow to be the value minus the cost of the flow.

Like the definitions for utility and passenger-taxicab equilibria in the continuous case, one can define them for the discrete case:
The utility of a taxicab $t_j\in s_u$
moving from $u$ to $v$, given new supply $s'$, surge prices $r$ and demand functions $d$, is
$$
\mu_{t_j}(u\mapsto v|s',r,d)=\frac{\min(|d_v(r_v)|,|s'_v|)}{|s'_v|}\cdot r_v-\ell_{u,v}.
$$

\begin{definition}\label{def:discrete.equi}
Given demand $d$ , current supply $s$ and new supply $s'$, surge prices $r$ are said to be in {\em passenger-taxicab} equilibrium if for every
min cost flow $f$ from $s$ to $s'$ and  for any
$u,v$ such that $f(u,v)>0$ we have that
\begin{itemize}
\item Taxicabs are choosing a best response: $\mu_{t_j}(u\mapsto v|s',r,d)=\max_{w\in V}(\mu_{t_j}(u\mapsto w|s',r,d))$.
\item All passengers $b\in B$ with $\val(b)>r_{\loc(b)}$ are served. No passengers $b\in B$ with $\val(b)<r_{\loc(b)}$ are served.
\end{itemize}
\end{definition}

\subinput{Online Setting}

In the online setting we inherit the continuous model setting, adding a function of time.
Time progresses in discrete time steps $1, 2, \ldots , T$.  At time
$t$ the demand vector $d^t=(d^t_1,d^t_2,\ldots,d^t_k)$ associates
each vertex $v\in V$ with some demand $d^t_v\geq 0$, and we assume
that the total demand $\sum_i d^t_i=1$. One should not think of a time step as being instantaneous, but rather as a period of time during which the demands remain steady.

Every time step $t$ also has an
associated supply vector $s^t=(s^t_1,s^t_2, \ldots, s^t_k)$, where
$s^t_i\geq 0$ and $\sum_i s^t_i =1$ for all $t$. The supply at time
$t$ is a ``reshuffle" of the supply at time $t-1$, by having infintestimally small
suppliers moving about the network.
In our model, the time required for suppliers to adjust supply from $s^{t-1}$ to $s^t$ is small relative to the period of time during which demand $d^t$ is valid.

If the demand in vertex $i$ at time $t$ is $d^t_i$, and the supply in
vertex $i$ at time $t$ is $s^t_i$, then the minimum of the two is the
actual demand served (in vertex $i$ at time $t$). Note that if the two
are not identical then there are either unhappy customers (without
service) or unhappy suppliers (with no customer to service).
Formally, we define the benefit derived during each time period, the {\sl demand served}, as in the continuous model.

We define the social welfare as follows:

\begin{definition}\label{def:socialwelfare}
	Given a demand sequence $d=(d^1, \ldots, d^T)$ and a supply sequence
	$s=(s^1, \ldots, s^T)$ we define the social welfare
	$$\sw\left(s,d\right) = \ds(s,d) - \ermv(s) = \sum_{t=1}^T
	\ds(s^t,d^t) - \sum_{t=2}^T \ermv(s^{t-1},s^t).$$
\end{definition}

An online algorithm for social welfare follows the following
structure. At time $t = 1, 2, \ldots, T$:

\begin{enumerate}
	\item A new demand vector $d^t$ appears.
	\item The online algorithm determines what the supply vector $s^t$ should be. (Indirectly, by computing and posting surge prices so that the resulting passenger-taxicab-equilibrium induces supply $s^t$).
\end{enumerate}

The goal of the online algorithm is to maximize the social welfare
as given in Definition \ref{def:socialwelfare}: Compute a supply
sequence $s$, so as to maximize $\sw(s,d)$. The supply vector $s^t$
is a function of the demand vectors  $d^1, \ldots, d^t$ but not of
any demand vector $d^\tau$, for $\tau>t$.
Implicitly, we assume that the passenger-taxicab equilibrium is attained quickly relative to the rate at which demand changes.

The competitive ratio of such an online algorithm, $\mbox{\rm Alg}$, is the worst case ratio
between the numerator: the social welfare resulting from the demand sequence $d$ and the
online supply $\mbox{\rm Alg}(d)$, and the denominator: the optimal social
welfare for the same demand sequence, {\sl i.e.},
$$
\min_{d} \frac{\sw(\mbox{\rm Alg}(d),d)}{\max_s\sw(s,d)}.
$$

\section{The Continuous Passenger-Taxicab Setting}\label{sec:surge}

In this section we deal with the continuous passenger-taxicab setting.
Given current supply $s$, demand $d$ and new supply $s'=d$, we show how to set surge prices $r$
such that they are in passenger-taxicab equilibria. Moreover, for these $s$, $d$, and $r$, the only possible $s'$ which results in a passenger-taxicab
equilibria is $s'=d$.
(Similar techniques give surge prices that induce [almost] arbitrary supply vectors, $\tilde{s}$, see below).

 Proof overview: Given some min cost flow $f^*$ from
supply  $s$ to demand $d$, we construct a unit demand market, with bidders and items. For every $x,y$ such that $f^*(x,y)>0$ we construct a bidder and an item. We also define bidder valuations for all items.
This unit demand market has Walrasian clearing prices that maximize social
welfare (Lemma~\ref{lemma:properallocation}). We show how we can
convert the Walrasian prices on items to surge pricing
(Lemma~\ref{lemma:Wprice-location}).

We then show and that the resulting surge pricing has a passenger-taxicab equilibrium which induces supply
equals demand (Lemma~\ref{lemma:exists_eq}) and it is the case with
all all passenger-taxicab equilibria (Lemma~\ref{lemma:eq_unique}).
Lemma \ref{lemma:altflow} shows that the incentive requirements in Equation
(\ref{eq:incentivesequilibrium}) also hold for any min cost flow $f\neq f^*$, from $s$ to $d$.
This proves
{\sl Theorem \ref{thm:surge}}.

\begin{figure}
\vbox{\hbox{
\begin{minipage}{0.49\textwidth}
\centering
\begin{tikzpicture}[shorten >=1pt, auto, node distance=3cm, ultra thick,
node_style/.style={circle,draw=blue,fill=blue!20!,font=\sffamily\Large\bfseries},
edge_style/.style={draw=black, ultra thick},scale=0.7]

\node[node_style] (v1) at (-4,2) {$v_1$}; \node[node_style] (v2) at
(0,2) {$v_2$}; \node[node_style] (v3) at (4,2) {$v_3$};
\node[node_style] (v4) at (-4,-2) {$v_4$}; \node[node_style] (v5) at
(0,-2) {$v_5$}; \node[node_style] (v6) at (4,-2) {$v_6$};

\draw[edge_style] (v1) edge node{1} (v2); \draw[edge_style] (v2)
edge node{1} (v3); \draw[edge_style] (v4) edge node{1} (v5);
\draw[edge_style] (v5) edge node{1} (v6); \draw[edge_style] (v1)
edge[bend left] node{2} (v3); \draw[edge_style] (v4) edge[bend
right] node{2} (v6);

\draw[edge_style] (v1) edge node[near start]{1} (v4);
\draw[edge_style] (v5) edge node[near end]{2} (v1);
\draw[edge_style] (v1) edge node[very near start]{3} (v6);

\draw[edge_style] (v4) edge node[very near end]{2} (v2);
\draw[edge_style] (v2) edge node[near start]{1} (v5);
\draw[edge_style] (v2) edge node[very near start]{2} (v6);

\draw[edge_style] (v4) edge node[very near end]{3} (v3);
\draw[edge_style] (v3) edge node[near start]{2} (v5);
\draw[edge_style] (v3) edge node[near start]{1} (v6);

\end{tikzpicture}

\captionof{figure}{Example road network, with costs along edges.}
\label{fig:example}\end{minipage}

\begin{minipage}{0.49\textwidth}
\centering
\begin{tikzpicture}[shorten >=1pt, auto, node distance=3cm, ultra thick,
node_style/.style={circle,draw=blue,fill=blue!20!,font=\sffamily\Large\bfseries},
edge_style/.style={draw=black, ultra thick},
dedge_style/.style={->,> = latex}, scale=0.6]

\node[node_style] (v1) at (-4,2) {$v_1$}; \node[node_style] (v2) at
(0,2) {$v_2$}; \node[node_style] (v3) at (4,2) {$v_3$};
\node[node_style] (v4) at (-4,-2) {$v_4$}; \node[node_style] (v5) at
(0,-2) {$v_5$}; \node[node_style] (v6) at (4,-2) {$v_6$};

\draw[dedge_style] (v3) edge[loop right] node{$\frac{1}{8}$} (v3);
\draw[dedge_style] (v3) edge node[near start]{$\frac{1}{8}$} (v6);
\draw[dedge_style] (v3) edge node[near start]{$\frac{1}{12}$} (v5);

\draw[dedge_style] (v2) edge node[near start]{$\frac{7}{24}$} (v5);
\draw[dedge_style] (v2) edge node[near
start,left=10pt]{$\frac{1}{24}$} (v4);

\draw[dedge_style] (v1) edge node[near start,left]{$\frac{1}{3}$}
(v4);

\end{tikzpicture}
\captionof{figure}{A min earthmover cost flow from supply vector
$\cursupply=\langle\frac{1}{3},\frac{1}{3},\frac{1}{3},0,0,0\rangle$ to
demand vector
$\newdemand=\langle0,0,\frac{1}{8},\frac{3}{8},\frac{3}{8},\frac{1}{8}\rangle$.}
\label{fig:flow1}
\end{minipage}
}}
\end{figure}

\begin{figure}
\vbox{\hbox{
\begin{minipage}{0.49\textwidth}
\begin{tikzpicture}[shorten >=1pt, auto, node distance=3cm, ultra thick,
node_style/.style={circle,draw=blue,fill=blue!20!,font=\sffamily\Large\bfseries},
edge_style/.style={draw=black, ultra thick},
dedge_style/.style={->,> = latex},scale=0.6]

\node[node_style] (v1) at (-4,2) {$v_1$}; \node[node_style] (v2) at
(0,2) {$v_2$}; \node[node_style] (v3) at (4,2) {$v_3$};
\node[node_style] (v4) at (-4,-2) {$v_4$}; \node[node_style] (v5) at
(0,-2) {$v_5$}; \node[node_style] (v6) at (4,-2) {$v_6$};

\draw[dedge_style] (v3) edge[loop right] node{$\frac{1}{8}$} (v3);
\draw[dedge_style] (v3) edge node[near start]{$\frac{1}{8}$} (v6);
\draw[dedge_style] (v3) edge node[near start]{$\frac{1}{24}$} (v5);
\draw[dedge_style] (v3) edge node[near start, above]{$\frac{1}{24}$}
(v4);

\draw[dedge_style] (v2) edge node[near start]{$\frac{1}{3}$} (v5);

\draw[dedge_style] (v1) edge node[near start, left]{$\frac{1}{3}$}
(v4);

\end{tikzpicture}
\captionof{figure}{Another min earthmover cost flow from from supply
vector
$\cursupply=\langle\frac{1}{3},\frac{1}{3},\frac{1}{3},0,0,0\rangle$ to
demand vector
$\newdemand=\langle0,0,\frac{1}{8},\frac{3}{8},\frac{3}{8},\frac{1}{8}\rangle$.}
\label{fig:flow2}
\end{minipage}
\begin{minipage}{0.49\textwidth} \centering
        \begin{tabular}{|c|c|c|c|c|c|c|}
            \hline
            &  $m_{14}$&  $m_{24}$&  $m_{25}$&  $m_{33}$&  $m_{35}$&  $m_{36}$\\
            \hline
            $b_{14}$&  3&  3&  2&  2&  2&  1\\
            \hline
            $b_{24}$&  2&  2&  3&  3&  3&  2\\
            \hline
            $b_{25}$&  2&  2&  3&  3&  3&  2\\
            \hline
            $b_{33}$&  1&  1&  2&  4&  2&  3\\
            \hline
            $b_{35}$&  1&  1&  2&  4&  2&  3\\
            \hline
            $b_{36}$&  1&  1&  2&  4&  2&  3\\
            \hline
            $p_{ij}$&  0&  0&  1&  3&  1&  2\\
            \hline
        \end{tabular}
            \captionof{figure}{Item valuations for bidders $B^f=\{b_{14},b_{24},b_{25},b_{33},b_{35},b_{36}\}$, items $M^f=\{m_{14},m_{24},m_{25},m_{33}, m_{35},m_{36}\}$, where $f$ is the min earthmover flow given in Figure \ref{fig:flow1}. Note that in Figure \ref{fig:example} we have $\max_{ij}\left(\ell_{ij}\right)=3$ and thus $C=4$. The last row gives Walrasian market clearing prices for items $m_{ij}$. Note that $p_{ij}=p_{ij'}$ for all $b_{ij}, b_{ij'}\in B^f$. }
    \label{fig:marketfig}
\end{minipage}}}

\end{figure}

As a running example, consider the road network in Figure
\ref{fig:example}. Also, assume that the supply vector
$\cursupply=\langle\frac{1}{3},\frac{1}{3},\frac{1}{3},0,0,0\rangle$
and demand vector
$\newdemand=\langle0,0,\frac{1}{8},\frac{3}{8},\frac{3}{8},\frac{1}{8}\rangle$.
Two minimum cost flows are given in Figures
\ref{fig:flow1} and \ref{fig:flow2}. Both these flows have
cost $1$.

Given a minimum cost flow $f^*$, % for each $v\in V$ let $j_v=|\{u | f(v,u)>0\}|$ and $i_v=|\{u | f(u,v)>0\}|$.
we define a unit demand market setting as follows:
\begin{itemize}
\item Items $M^{f^*}$, and unit demand bidders $B^{f^*}$, both of which are indexed by pairs of vertices, where
$$M^{f^*}= \left\{ m_{xy} | x,y\in V, f^*(x,y)>0 \right\} \qquad B^{f^*}= \left\{ b_{wz} | w,z\in V, f^*(w,z)>0 \right\}.$$
\item We set the value of item $m_{xy}\in M^{f^*}$ to bidder $b_{wz}\in B^{f^*}$ to be, $$\auctionvalue_{b_{wz}}(m_{xy})=C - \ell_{w,y}, \qquad\mbox{\rm where\ } C=\max_{i,j}{\ell_{i,j}+1}.$$
\item The utilities of bidders are unit demand and quasi-linear, {\sl i.e.}, the utility $\auctionutility_{b_{wz}}$ of bidder $b_{wz}\in B^{f^*}$ for item set $S$ and price $p$ is
$$ \auctionutility_{b_{wz}}\left(S\right)=\max_{m_{xy}\in S} \auctionvalue_{b_{wz}}(m_{xy})-p.$$
\end{itemize}

As an example, let $f^*$ be the minimum cost flow of Figure~\ref{fig:flow1}.
The market induced by $f^*$ is illustrated in Figure~\ref{fig:marketfig}.

Given a flow $f^*$, bidders $B^{f^*}$ and items $M^{f^*}$ we define
the following weighted bipartite graph $G(B^{f^*},M^{f^*},E)$, where
between bidder $b_{wz}\in B^{f^*}$ and item $m_{xy}\in M^{f^*}$
there is an edge of weight $C-\ell_{w,y}\geq 1$.
%

%The maximum weight matching is a perfect matching where each bidders is allocated a unique item
%(and no bidder is assigned to the $\emptyset$).

\begin{definition}
Given a flow $f^*$, a {\em matching} between bidders $B^{f^*}$ and
items $M^{f^*}$ is a function $\pi:B^{f^*}\mapsto M^{f^*} \cup
\{\emptyset\}$, where bidder $b\in B^{f^*}$ is matched to item
$\pi(b)\in M^{f^*}$ or unmatched (if $\pi(b)=\emptyset$), such that
no two bidders $b_1,b_2\in B^{f^*}$ are matched to the same item
$m\in M^{f^*}$.

As there is an edge between every bidder $b_{wz}$ and every item
$m_{xy}$ with weight $C-\ell_{w,y}\geq 1$, the maximum weight
matching is a perfect matching between bidders and items and the
mapping $\pi$ never assigns $\emptyset$ to a bidder.
\end{definition}

\begin{figure}
\begin{minipage}{0.49\textwidth}
\centering
        \begin{tabular}{|c|c|c|c|c|c|c|}
            \hline
            &  $v_1$&  $v_2$&  $v_3$&  $v_4$&  $v_5$&  $v_6$\\
            \hline
            $r^t_i$&  1&1&1&4&3&2\\
            \hline

        \end{tabular}
            \captionof{figure}{Surge prices resulting in a flow-equilibrium with $s^t=d^t$. These surge price for $v_j$ is $C-p_{ij}$ if there exists some bidder $b_{ij}\in B^f$ and 1 otherwise. The Walrasian prices $p_{ij}$ appear in Figure \ref{fig:marketfig}. }
    \label{fig:surgeprice}
\end{minipage}
\end{figure}

\begin{lemma}\label{lemma:properallocation}
The matching $g$ where $g(b_{wz})=m_{wz}$,
% $b_{wz}\in B^f$, $m_{wz}\in M^f$,  (allocates item $m_{wz}$ to bidder $b_{wz}$),
maximizes social welfare.
In addition, there exist
Walrasian prices for which $g$ is a competitive market equilibrium.
\end{lemma}

\begin{proof}
The proof is via contradiction.
Assume there exists some matching $\tilde{g}:B^{f^*}\mapsto M^{f^*}$
with strictly greater social welfare than the matching $g$. For a
bidder $b\in B^{f^*}$, define $\tilde{h}(b)=z$, iff
$\tilde{g}(b)=m_{wz}$ for some $w\in V$, and
%Let $g$ be the matching $g(b_{xy})=m_{xy}$, $b_{xy}\in B^f$, $m_{xy}\in M^f$. Also,
$h(b)=z$, iff $g(b)=m_{wz}$ for some $w\in V$. Note that
$h\left(b_{wu}\right)=u$ so for a given $w$ and $z$ we have
$\left|\left\{u| h\left(b_{wu}\right)=z\right\}\right|=1$ if
$f(w,z)>0$ and zero otherwise.

Choose $\epsilon$ to be the minimum non-zero flow in $f^*$, {\sl i.e.},
$\epsilon = \min\{f^*(w,z)|f^*(w,z)>0\}$. We now
define a flow $f'$, which is a slight perturbation of flow $f^*$. In
flow $f'$, the flow from $w$ to $z$ is:
$$f'(w,z)= f^*(w,z) + \epsilon\left(\left|\left\{u| \tilde{h}\left(b_{wu}\right)=z\right\}\right| - \left|\left\{u| h\left(b_{wu}\right)=z\right\}\right|\right).$$

We first prove that $f'$ is a valid flow, and later we show
that it has a lower cost than $f^*$, in
contradiction to the minimality of $f^*$.

%(The proof of the following lemma is deferred to Appendix \ref{sec:appmissing}.)
\begin{lemma}\label{lem:validflow}
Flow $f'$ is a valid flow from supply vector $\cursupply$ to demand
vector $\newdemand$.
\end{lemma}

\begin{proof}
    Consider the requirements that $f'$ be a valid flow:
    \begin{itemize}
        \item For all $x,y\in V$, $f'\left(x,y\right)\geq 0$ : By definition of $f'$ if $f^*(w,z)=0$ then $f'(w,z)\geq 0$ and if $f^*(w,z)>0$ then $f'(w,z)\geq f^*(w,z)-\min\{f^*(w,z)|f(w,z)>0\}\geq 0$.
        \item For all $x\in V$, $\sum_y f'\left(x,y\right)=\cursupply_x$ : By definition of $f^*$ we have $\sum_y f^*(x,y)=\cursupply_x$. Thus,
        \begin{eqnarray*}\sum_y f'(x,y)&=& \sum_y \bigg(f^*(x,y) + \big(\left|\left\{u| \tilde{h}\left(b_{xu}\right)=y\right\}\right| - \left|\left\{u| h\left(b_{xu}\right)=y\right\}\right|\big)\cdot\epsilon\bigg) \\
            &=&\cursupply_x + \left(\sum_y \left|\left\{u| \tilde{h}\left(b_{xu}\right)=y\right\}\right| - \sum_y \left|\left\{u| h\left(b_{xu}\right)=y\right\}\right|\right)\cdot\epsilon \\
            &=& \cursupply_x + \left(\left|\left\{u|b_{xu}\in B^f\right\}\right|-\left|\left\{u|b_{xu}\in B^f\right\}\right|\right)\cdot\epsilon\\
            &=&\cursupply_x.\end{eqnarray*}
        \item For all $y\in V$, $\sum_x f'\left(x,y\right)=\newdemand_y$ : By definition of $f^*$ we have $\sum_x f^*(x,y)=\newdemand_y$. Thus,
        \begin{eqnarray*}\sum_x f'(x,y)&=& \sum_x \bigg(f^*(x,y) + \big(\left|\left\{u| \tilde{h}\left(b_{uy}\right)=x\right\}\right| - \left|\left\{u| h\left(b_{uy}\right)=x\right\}\right|\big)\cdot\epsilon\bigg) \\
            &=&\newdemand_y + \left(\sum_x \left|\left\{u| \tilde{h}\left(b_{uy}\right)=x\right\}\right| - \sum_x \left|\left\{u| h\left(b_{uy}\right)=x\right\}\right|\right)\cdot\epsilon \\
            &=& \newdemand_y + \left(\left|\left\{u|b_{uy}\in B^f\right\}\right|-\left|\left\{u|b_{uy}\in B^f\right\}\right|\right)\cdot\epsilon\\
            &=&\newdemand_y.\end{eqnarray*}
    \end{itemize}
\end{proof}

From the fact that $\tilde{g}$ has a higher social welfare we get,
\[
\sum_{w,z:b_{wz}\in B^{f^*}} \auctionvalue_{b_{wz}}(\tilde{g}(b_{wz})) >
\sum_{w,z:b_{wz}\in B^{f^*}} \auctionvalue_{b_{wz}}({g}(b_{wz}))\;.
\]
Using the definition of the valuations we have,
\[
\sum_{w,z:b_{wz}\in B^{f^*}} C - \ell_{w,\tilde{h}(b_{wz})} >
\sum_{w,z:b_{wz}\in B^{f^*}} C - \ell_{w,{h}(b_{wz})}\;.
\]
This implies that
\[
\sum_{w,z:b_{wz}\in B^{f^*}} \ell_{w,{h}(b_{wz})} >
\sum_{w,z:b_{wz}\in B^{f^*}} \ell_{w,\tilde{h}(b_{wz})}\;.
%\label{eqn:diffflow}
\]
Using this last inequality, it follows
%        It follows from Equation \eqref{eqn:diffflow}
that the  cost of $f^*$ (Definition \ref{def:earthmover}) satisfies
\begin{align*}
\ermv(f^*) &= \sum_{x,y} f^*(x,y)\cdot\ell_{x,y} \\
&> \sum_{x,y} f^*(x,y)\cdot\ell_{x,y} + \left(\sum_{w,z:b_{wz}\in
B^f} \ell_{w,\tilde{h}(b_{wz})} - \sum_{w,z:b_{wz}\in B^f}
\ell_{w,{h}(b_{wz})}\right)\cdot\epsilon = \ermv(f'),
\end{align*}
which contradicts the fact that flow $f^*$ is a minimum cost
flow.

The fact that there exist Walrasian prices for $g$ that are in
competitive market equilibrium follows from \cite{GS99}. This concludes the proof of Lemma \ref{lemma:properallocation}.
\end{proof}

Let the Walrasian price of $m_{xy}$ be $p_{xy}$ as guaranteed by the
lemma above. We first show that any two prices which
correspond to the same vertex must have the same price.

\begin{lemma}
\label{lemma:Wprice-location}
 %       For the market setting above any Walrasian pricing,
For any two items $m_{xy}$ and $m_{x'y}$ we have $p_{xy}=p_{x'y}$.
\end{lemma}

\begin{proof}
%        If not, this contradicts the envy free properties of Walrasian prices, every bidder must get an item in the demand set for that bidder.
For contradiction assume that $p_{xy}>p_{x'y}$. Let $b_{wz}$ be the
bidder assigned $m_{xy}$. Thus, for item $m_{xy}$ bidder $b_{wz}$
has utility
$\auctionutility_{b_{wz}}(m_{xy})=C-\ell_{w,y}-p_{xy}<C-\ell_{w,y}-p_{x'y}=\auctionutility_{b_{wz}}(m_{x'y})$
which implies that $m_{xy}$ is not in the demand set for bidder
$b_{wz}$. A contradiction to the fact that $p$ are Walrasian prices.
\end{proof}

For any $y\in V$ such that there exist items of the form $m_{xy}$
for some $x\in V$, let $p_y$ denote the Walrasian price for such
items (By Lemma~\ref{lemma:Wprice-location} all those Walrasian prices are
identical). If no items of the form $m_{xy}$ exist, this implies
that demand at vertex $y$, $\newdemand_y=0$, and we can set $p_y=0$.
Define surge prices, $\surgevec_y = C-p_y$, for all $y\in V$.

\begin{lemma}
\label{lemma:exists_eq}
Given current supply $\cursupply$ and demand $\newdemand$, surge prices $\surgevec_y=C-p_y$, new support $s'=d$, and $x,y,w\in V$
 such that $f^*(x,y)>0$ then
 $$\flowutility\left(x\rightarrow y\right|s',r,d)\geq \flowutility\left(x\rightarrow w\right|s',r,d).$$
\end{lemma}

\begin{proof}
    Let $x,y$ be such that $f^*(x,y)>0$.% with surge prices $\surgevec_y$.
    %generated via the market equilibrium between $B^f$ and $M^f$, then:
    Then,
    \begin{eqnarray}
    \flowutility\left(x\rightarrow y\right|s',r,d) &=& \surgevec_y\cdot\min{\left(1,\frac{\newdemand_y}{\newsupply_y}\right)}-\ell_{x,y} \label{eq:firstdef}\\ %By definition
    &=& \surgevec_y-\ell_{x,y} \label{eq:minone}\\ %This is because the equilibrium sets s^t=d^t and thus the minimum is alway 1
    &=& C-p_y-\ell_{x,y} \label{eq:seconddef}\\ %By the definition of the surge prices
    &=& \auctionutility_{b_{xy}}\left(m_{xy}\right) \label{eq:thirddef}\\ %This is the utility in the AUCTION
    &\geq& \auctionutility_{b_{xy}}\left(m_{zw}\right)\qquad \forall m_{zw}\in M^{f^*} \Leftrightarrow \forall m_{zw} : s'_w>0 \label{eq:marketeq}\\ %This is a direct derivation from the walrasian pricing
    &=& C-p_w-\ell_{x,w}  \label{eq:thirddefmirror}\\
    &=& \surgevec_w-\ell_{x,w} \label{eq:seconddefmirror}\\
    &\geq& \surgevec_w\cdot\min{\left(1,\frac{\newdemand_w}{\newsupply_w}\right)}-\ell_{x,w} \label{eq:secondminonw}\\
    &=& \flowutility\left(x\rightarrow w\right|s',r,d) \label{eq:firstdefmirror}%This is again the utility as defined by the flow function
    \end{eqnarray}

    Equations \eqref{eq:firstdef},\eqref{eq:firstdefmirror} follow from
    the definition of the utility in the continuous passenger-taxicab setting, definition \ref{def:utility}.\\
    Equations \eqref{eq:minone},\eqref{eq:secondminonw} follows from
    considering the passenger-taxicab equilibrium where $\newsupply=\newdemand$ resulting in
    $\frac{\newdemand_y}{\newsupply_y}=1$ for all $y$.\\
    Equations \eqref{eq:seconddef},\eqref{eq:seconddefmirror} follow
    from the definition of the surge prices.\\
    Equations \eqref{eq:thirddef},\eqref{eq:thirddefmirror} follow from
    the definition of the utility in the market setting.\\
    Equation \eqref{eq:marketeq} follows from the market equilibrium.
    %This proves the equilibrium for the flow functio

    So, we have that $$\flowutility\left(x\rightarrow y\right|s',r,d) \geq \flowutility\left(x\rightarrow w\right|s',r,d) \qquad \forall w : s'_w>0.$$
    It remains to consider $\flowutility\left(x\rightarrow w\right|s',r,d)$ for $w$ such that $s'_w=0$. In this case the surge price at $w$ is zero,
    so the utility $\flowutility\left(x\rightarrow w\right|s',r,d)\leq 0$.
\end{proof}

The following lemma shows that the incentive requirements of Equation (\ref{eq:incentivesequilibrium}) hold,
not only for the flow $f^*$, but also for {\sl any} min cost flow from $s$ to $d$.
\begin{lemma}
\label{lemma:altflow}
Fix current supply $\cursupply$ and demand $\newdemand$, surge prices $\surgevec_y=C-p_y$, and new support $s'=d$. Let $f'$ be an arbitrary min
cost flow from $s$ to $s'=d$, then, for any $x,y,w\in V$
 such that $f'(x,y)>0$ we have that
 $$\flowutility\left(x\rightarrow y\right|s',r,d)\geq \flowutility\left(x\rightarrow w\right|s',r,d).$$
\end{lemma}
\begin{proof}
Define $$\Gamma(f) = \sum_{u\in V}\sum_{v\in V} f(u,v)\cdot(r_v - \ell_{u,v}) = \sum_{v\in V} s'_v\cdot r_v - \ermv(s,s').$$
As $f'$ and $f^*$ are both min cost flows from $s$ to $s'$ we have that $\Gamma(f^*)=\Gamma(f')$.

%Let $f'$ be some min cost flow from $s$ to $s'=d$. In
For contradiction assume there exist some $u,v$ such that $f'(u,v)>0$ and
$\mu(u\mapsto v|s',r,d)<\max_{w\in V}\mu(u\mapsto w|s',r,d)$.
\begin{eqnarray}
\Gamma(f^*) &=& \sum_{u\in V} \sum_{v\in V} f^*(u,v)\cdot(r_v-\ell_{u,v})  \label{eq:gamma1}\\
&=& \sum_{u\in V} \sum_{v\in V} f^*(u,v)\cdot\max_{v\in V}(r_v-\ell_{u,v}) \label{eq:gamma2}\\
&=&\sum_{u\in V} s_u\cdot\max_{v\in V}(r_v-\ell_{u,v}) \label{eq:gamma3}\\
 &=& \sum_{u\in V} \sum_{v\in V} f'(u,v)\cdot\max_{v\in V}(r_v-\ell_{u,v}) \label{eq:gamma4}\\
&>& \sum_{u\in V} \sum_{v\in V} f'(u,v)\cdot(r_v-\ell_{u,v}) = \Gamma(f')\;. \label{eq:gamma5}
\end{eqnarray}
Eq. (\ref{eq:gamma1}) follows from the definition of $\Gamma$. Eq. (\ref{eq:gamma2}) follows from
Lemma~\ref{lemma:exists_eq}. Eq. (\ref{eq:gamma3}) and (\ref{eq:gamma4}) follows from the definition of a flow,
since for any flow $f$ from $s$ we have that $s_u=\sum_{v\in V}f(u,v)$.
Eq. (\ref{eq:gamma5}) holds since we assumed, for contradiction, that there exist some $u,v$ such that $f'(u,v)>0$ and
$\mu(u\mapsto v|s',r,d)<\max_{w\in V}\mu(u\mapsto w|s',r,d)$.
Hence, we reached a contradiction to the assumption that $f'$ is a min cost flow.
\end{proof}

If follows from the Lemma above that computing surge prices $r$ via flow $f^*$ ensures that
taxicab routing using any other min cost flow $f'$ is also a best response under surge prices $r$.

Next we show that all relevant passenger-taxicab equilibria all have new supply $\newsupply=\newdemand$.

\begin{lemma}
\label{lemma:eq_unique} Given current supply $s$, demand $d$, and surge prices $\surgevec_y=C-p_y$,
all passenger-taxicab equilibria induce $\newsupply=d$.
\end{lemma}

%Due to  lack of space, the proof of this lemma is deferred to Appendix \ref{sec:appmissing}.

\begin{proof}
    Let $f$ be a min cost flow from $s$ to $d$. For
    contradiction, assume that that there exists some $\bar{s}\neq d$ such that $s,s'=\bar{s},d,r$ are in a passenger-taxicab equilibrium. Let $f'$
    be some min cost flow from $s$ to $\bar{s}$.

    Consider $H=\left\{y|
    \sum_x f'(x,y)>\newdemand_y\right\}$ ({\sl i.e.}, the set of all vertices for
    which the  flow $f'$ results in strictly more supply
    than demand). Since $s'\neq d$ and they both sum to 1, we have that $H\neq \emptyset$.  Let $H'=\left\{x| \exists y\in H \, \mathrm{s.t.} \,
    f'(x,y)>0 \right\}$ ({\sl i.e.}, the set of all vertices from which
    supply flows to $H$). As $H\neq \emptyset$ it follows that $H'\neq \emptyset$.

    We claim that there exists some $w\in H'$, $y\not\in H$, such that
    $f(w,y)>0$. For contradiction assume that all the flow in $f$ from
    vertices in $H'$ is to vertices in $H$. By definition of flows: $\sum_{y\in V}
    f\left(x,y\right) = \cursupply_x= \sum_{y\in V} f'\left(x,y\right)$. We now
    have,
    \begin{eqnarray*}
         \sum_{x\in H'}\sum_{y\in V} f(x,y) &=& \sum_{x\in H'}\sum_{y\in H} f(x,y) \\
         &=& \sum_{x\in H'}\sum_{y\in H} f'(x,y) \\ &=&\sum_{y\in H} \sum_{x\in H'} f'(x,y) \\ %This is directly from the above statement about flows in general
        &>& \sum_{y\in H} \newdemand_y \\ %This is due to the fact that we sum up all of the inputs to every x in H
        &=& \sum_{x\in H'}\sum_{y\in H} f(x,y), %This is from the contradiction assumption
    \end{eqnarray*}
    which is a contradiction.

    This implies that there exist  $w\in H'$, $x\not\in
    H$ such that $f\left(w,x\right)>0$. Since $w\in H'$ there also exists some $y\in H$ such that
    $f'(w,y)>0$.

    We have shown that $s,s'=d,d,r$ is in passenger-taxicab equilibrium,
    this implies that $$\flowutility(w\mapsto x|s',r,d)=\surgevec_x - \ell_{w,x}
    \geq \surgevec_y - \ell_{w,y}=\flowutility(w\mapsto y|s',r,d).$$ Since $y\in H$ we have
    we have that $\sum_u f'\left(u,y\right)>\newdemand_y$ resulting in the utility
    $$
    \flowutility\left(w\mapsto y\right|s',r,d)= \left(\surgevec_y -
    \ell_{w,y}\right)\cdot\min\left(1,\frac{\newdemand_y}{\newsupply_y}\right)<\surgevec_y-\ell_{w,y}
    \leq \surgevec_x-\ell_{w,x} = \flowutility\left(w\mapsto x|s',r,d\right),
    $$
    where $x\not\in H$ implies the last equality. This is in contradiction
    to the $s,s'=\bar{s},d,r$ being a passenger-taxicab equilibrium.
\end{proof}

Theorem \ref{thm:surge} follows from Lemma~\ref{lemma:exists_eq}, Lemma~\ref{lemma:altflow}, and
Lemma~\ref{lemma:eq_unique}.

\begin{theorem}\label{thm:surge}
Given distances $\ell_{i,j}$ and an arbitrary supply vector, $\cursupply=\langle \cursupply_1, \ldots, \cursupply_k \rangle$.
Let the demand vector be $\newdemand=\langle \newdemand_1, \ldots,
\newdemand_k\rangle$. Then, there exists a surge price vector $\surgevec=\langle
\surgevec_1,\ldots,\surgevec_k\rangle$ that results in a passenger-texicab
equilibrium which induces a supply $\newsupply=\newdemand$.
Moreover, any passenger-taxicab equilibrium of $\surgevec$ induces supply
$\newsupply=\newdemand$, and the surge prices $\surgevec$ can be
computed in polynomial time.
\end{theorem}

We can extend the result from equating supply and demand to
modifying the supply vector $\cursupply$ to any supply $\newsupply$, with
the restriction that if $\newsupply_i>0$ then $\newdemand_i>0$. The new surge
prices are computed as follows. First we compute, as before, the
surge prices $\surgevec$ from $\cursupply$ to $\newdemand$. Then, we set
$\bar{\surgevec}_i=\max\{1,\frac{\newsupply_i}{\newdemand_i}\}\surgevec_i$ and the
resulting surge prices are $\bar{\surgevec}$. In a similar way we can
establish,

\begin{theorem}\label{thm:anyst}
Let $\newdemand=\langle \newdemand_1,\ldots, \newdemand_k\rangle$
and let $\alpha=\langle \alpha_1,\ldots, \alpha_k\rangle$ be the
target supply vector, subject to the restriction that if
$\alpha_i>0$ then $\newdemand_i>0$. Then there exist surge prices
$\bar{\surgevec}$ for which some passenger-taxicab equilibrium induces supply
$\alpha$.
\end{theorem}

\section{The Discrete Passenger-Taxicab Setting}\label{sec:respdemand}

In this section we consider a more realistic scenario where both demand and supply are sensitive to the surge pricing.
All else being equal, higher surge prices mean less demand and more supply.

We define social welfare to be the sum of valuations of passengers served minus the sum of the distances traversed by the taxicabs to serve these passengers (Definition \ref{def:sw}).
 Given current supply $s$ and a passenger profile $P$, we give an algorithm for computing surge prices $r$ that creates a passenger-taxicab equilibrium that maximizes social welfare.

The location of a passenger $b_i$ and taxi $\taxi_j$
is denoted by $\loc(\pass_i)$ and $\loc(\taxi_j)$ respectively ({\sl i.e.}, $\pass_i\in P_{\loc(\pass_i)}$ $\taxi_j\in s(\loc(\taxi_j))$).
For brevity, we use the notation $\overline{\ell}_{i,j}=\ell_{\loc(\pass_i),\loc(\taxi_j)}$.

% Also, for some location $\gamma$ denote $\dist(\gamma,\taxi_j)=\ell_{\gamma,j}$.

%
%OLD: Each passenger in node $w$ will have a valuation for the
%requested ride, w.l.o.g. assume
%$\pass_{w,1}\geq\pass_{w,2}\geq\pass_{w,3}\ldots$. In order to
%maximize social welfare, if there are $m$ taxis at node $w$, the
%served passengers are $\pass_{w,1},\ldots,\pass_{w,m}$. The total
%number of taxis is $n$ and denote by $\taxi_i$ the location of the
%$i$-th taxi.

%Assume then for each node $v$ there exists a demand curve function $f_v$

%\subinput{Results}
%In this section we show a polynomial algorithm based on the minimum cost flow problem that calculates the
%maximal social welfare regardless of prices. Next, we show that in some cases this maximal utility cannot
%be achieved for any prices. Lastly, %TODO

\subinput{Maximizing Social Welfare}

As in the continuous case, we reduce the problem of computing surge prices to computing market clearing prices
in a unit demand market.
Given a set of passengers $B$ and taxicabs $T$, we construct a unit demand market $M(B,T)$, where $B$ is the set of buyers and $T$ is the set of items. For the unit demand market, $M(B,T)$, we set the value of buyer $\pass_i\in B$ for item $\taxi_j\in T$  to be
$\auctionvalue_{\pass_i}(\taxi_j)=\val(\pass_i)-\overline{\ell}_{i,j}.$

%Note that for all allocations of items to buyers in $M(B,T)$,
%the social welfare in $M(B,T)$ (sum of buyer valuations) is equal to the social
%welfare defined in Equation \ref{eq:swdiscrete}.

Let the allocation where item $\taxi_j$ is given to $\buyer(\taxi_j)=\pass_i$ be a social welfare maximizing allocation in the unit demand market $M(B,T)$. Also, let $\buyer(\taxi_j)=\emptyset$ if item $\taxi_j$ is unallocated.

 This social welfare maximizing allocation in $M(B,T)$ translates into a flow $f^*$ for the discrete passenger-taxicab problem where $\taxi_j$ moves from $\loc(\taxi_j)$ to $\loc(\pass_i)$ if $\buyer(\taxi_j)=\pass_i$. Ergo,
 $$f^*(u,v) = \left\{\begin{array}{lr}
        |\{(i,j)|\pass_i\in P_v,\taxi_j\in s_u,\buyer(\taxi_j) =\pass_i\} |   &\mathrm{if\ }u\neq v, \\
        |\{(i,j)|\pass_i\in P_v,\taxi_j\in s_u,\buyer(\taxi_j) =\pass_i\} | +
        |\{j|\taxi_j\in s_u, \buyer(\taxi_j) =\emptyset\}|,  &\mathrm{if\ }u=v.
        \end{array}\right.$$

Let $s'$ be such that $s'_v=\sum_u f^*(u,v)$ for all $v\in V$. We say that the new supply $s'$ is {\sl induced} by $f^*$.
We now show that
\begin{lemma}
 The flow $f^*$ is a min cost flow from $s$ to $s'$.
\end{lemma}
\begin{proof}
  Assume that $f'$ is a flow from $s$ to $s'$ of strictly lower cost. As $f'$ is an integral flow it can be decomposed into a union of unit flows. This can be interpreted as an alternative allocation in the $M(B,T)$ unit demand market, with strictly higher social welfare. This is in contradiction to our construction.
\end{proof}

Choose the minimal Walrasian prices to clear the unit demand market $M(B,T)$. Such prices are also VCG prices
\cite{Leonard83}. Let the Walrasian price for item $\taxi_j$ be $\price_{\taxi_j}$.
We now define surge prices $r_v$, $v\in V$, for the discrete passenger-taxicab problem.
Specifically, for all $v\in V$, set \begin{equation} r_v=\min_{\taxi_j\in T}(\ell_{\loc(\taxi_j),v}+\price_{\taxi_j}).\label{eq:discretesurgeprices}\end{equation}

\begin{lemma}\label{lem:assignmentismax}
	Assigning $\taxi_j$ to serve passenger $\buyer(\taxi_j)$ is a social welfare maximizing allocation.
\end{lemma}
\begin{proof}
    First, we show that for any allocation of taxicabs to passengers in the taxicab-passenger setting there exists an allocation of items to buyers in the unit demand market $M(B,T)$ such that the social welfare is the same. Then, we show that for the allocation of items to buyers that maximizes the social welfare in the unit demand market there exists an allocation of taxicabs to passengers with the same social welfare.

    Fix an allocation of passengers to taxicabs, {\sl i.e.}, $\Phi:B\rightarrow T \cup \{\emptyset\}$ is a matching. Given the matching $\Phi$ we define an allocation $\Pi:B\rightarrow T \cup \{\emptyset\}$ in the unit demand market where $\Phi(b) = \Pi(b)$ for all $b\in B$.

    The social welfare of $\Phi$ in the taxicab-passenger setting is $\sum_{b\in B} (\val(b) - \ell_{\loc(b),\loc(\Phi(b))})I_{\Phi(b)\neq \emptyset}$. Similarly, the social welfare of $\Pi$ in the unit demand market setting is $\sum_{b\in B}\auctionvalue_b(\Pi(b))=\sum_{b\in B} (\val(b) - \ell_{\loc(b),\loc(\Pi(b))})I_{\Pi(b)\neq \emptyset}$. Since $\Phi(b)=\Pi(b)$ it follows that any allocation in the taxicab-passenger setting has a corresponding allocation in the unit demand market with the same social welfare.

    We now show that an allocation in the unit demand market that maximizes social welfare has a corresponding allocation in the passenger-taxicab setting that also maximizes social welfare.
    Denote the maximal allocation in the unit demand market by $\Pi_{\max}:B\rightarrow T \cup \{\emptyset\}$.
    Define the corresponding matching of passengers to taxicabs by $\Phi_{\max}:B\rightarrow T \cup \{\emptyset\}$, where $\Phi_{\max}(b)=\Pi_{\max}(b)$ for all $b\in B$ ($\Phi_{\max}$ is a matching since $\Pi_{\max}$ is a valid allocation in a unit demand market).

    Moreover, we need to show that higher valued passengers have priority over lower valued passengers at the same location. {\sl I.e.}, we need to show that for any two passengers, $b_1,b_2\in B$, such that $\loc(b_1)=\loc(b_2)$ and $\Phi_{\max}(b_1)\neq\emptyset$, $\Phi_{\max}(b_2)=\emptyset$ we have that $\val(b_1)\geq\val(b_2)$.

    Contrariwise, assume for some $b_1,b_2\in B$ we have that $\loc(b_1)=\loc(b_2)$,  $\Phi_{\max}(b_1)\neq\emptyset$, $\Phi_{\max}(b_2)=\emptyset$, but $\val(b_1)<\val(b_2)$. Define in the unit demand market then $\Pi':B\rightarrow T\cup\{\emptyset\}$ such that $\Pi'(b_1)=\emptyset$, $\Pi'(b_2)=\Phi_{\max}(b_1)$ and $\Pi'(b)=\Pi_{\max}(b)$ for all $b\notin \{b_1,b_2\}$. 
    
    We now show that the social welfare under $\Pi$ is strictly greater than the social welfare under $\Pi'$: 
    \begin{eqnarray*} \sum_{b\in B}(\auctionvalue_b(\Pi'(b)))&=&\sum_{b\in B,b\neq b_1,b_2}(\auctionvalue_b(\Pi'(b))) + \auctionvalue_{b_2}(\Phi_{\max}(b_1))\\
    &=&\sum_{b\in B,b\neq b_1,b_2}(\auctionvalue_b(\Pi_{\max}(b))) + \val(b_2)-\dist(\loc(b_2),\loc(\Phi_{\max}(b_1)))\\
    &>&\sum_{b\in B,b\neq b_1,b_2}(\auctionvalue_b(\Pi_{\max}(b))) + \val(b_1)-\dist(\loc(b_1),\loc(\Phi_{\max}(b_1)))\\
    &=& \sum_{b\in B}(\auctionvalue_b(\Pi_{\max}(b))).\end{eqnarray*}
    Thus, $\Pi'$ has strictly higher social welfare than $\Pi_{\max}$ in unit demand setting in contradiction to $\Pi_{\max}$ maximizing social welfare. Thus, $\Phi_{\max}$ is a valid allocation in the taxicab-passenger setting which maximizes the social welfare.
%	The social welfare of an allocation in the unit demand market is $$\sum_{\taxi_j:\buyer(\taxi_j)\neq\emptyset} \auctionvalue_{\buyer(\taxi_j)}(\taxi_j) =   \sum_{\taxi_j:\buyer(\taxi_j)\neq\emptyset} \val(\buyer(\taxi_j))-\ell_{\loc(\buyer(\taxi_j)),\loc(\taxi_j)}.$$
%	Let $r'$ be the all zero price vector, and let $s'$ be as above.
%The social welfare in the unit demand market allocation is \begin{eqnarray*} SW(s,s',r',d)&=&\dsv(s',d,r')-em(s,s')\\ &\geq& \sum_{\taxi_j:\buyer(\taxi_j)\neq\emptyset} \val(\buyer(\taxi_j))-\ell_{\loc(\buyer(\taxi_j)),\loc(\taxi_j)}\\ &=&\sum_{\taxi_j:\buyer(\taxi_j)\neq\emptyset} \auctionvalue_{\buyer(\taxi_j)}(\taxi_j).\end{eqnarray*}
\end{proof}

%we claim that
%(a)
%$f^*$ is a min cost flow from $s$ to $s'$.
%(b) $f^*$ is a flow equilibrium for $s$, $s'$, surge prices $r$ as defined above,
%(c) for all $\pass_i$, if $\val(\pass_i)>r_{\loc(\pass_i)}$ then $\pass_i$ is served. (Define in model).
%(d) the computation of the surge prices is incentive compatible for the passengers

\begin{lemma}\label{lem:achievesmin}
For any passenger $\pass_i$ such that $\pass_i=\buyer(\taxi_j)$ we have that
$$\overline{\ell}_{i,j}+\price_{\taxi_j}=\min_{\taxi_z\in T}(\overline{\ell}_{i,z}+\price_{\taxi_z})=r_{\loc(\pass_i)}.$$
\end{lemma}

\begin{proof}
Since $\pass_i=\buyer(\taxi_j)$, and $\price$ are Walrasian prices,  we have that buyer $\pass_i$ maximizes its utility $\auctionutility_{\pass_i}$. Ergo,
\begin{eqnarray*} \auctionutility_{\pass_i}(\taxi_j)&=&\max_{t_x\in T}(\auctionutility_{\pass_i}(\taxi_x))\\ &=&
\max_{t_x\in T}(\val(\pass_i)-\overline{\ell}_{i,x}-\price_{\taxi_x})\\ &=&\val(\pass_i)-\min_{t_x\in T}(\overline{\ell}_{i,x}+\price_{\taxi_x})\\
&=&\val(\pass_i)-r_{\loc(\pass_i)}.\end{eqnarray*}

As
$$\auctionutility_{\pass_i}(\taxi_j)=\val(\pass_i)-\overline{\ell}_{i,j}-\price_{\taxi_j}=\val(\pass_i) -r_{\loc(\pass_i)} $$
it follows that
$$\overline{\ell}_{i,j}+\price_{\taxi_j}=\min_{\taxi_x\in T}(\overline{\ell}_{i,x}+\price_{\taxi_x}).$$
\end{proof}

\begin{lemma} \label{lemma:envyfree}
Any passenger $\pass_i$  that is not served is not interested in being served (or is indifferent), {\sl i.e.}, then
$\val(\pass_i)\leq r_{\loc(\pass_i)}$. Any passenger $\pass_i$ that is served has $\val(\pass_i)\geq r_{\loc(\pass_i)}$.
\end{lemma}

\begin{proof}
Let  $\pass_i$ be some buyer allocated no item in the social welfare maximizing allocation for $M(B,T)$, then it must be that  $\max_{\taxi_x\in T}\auctionutility_{\pass_i}(\taxi_x)\leq 0$.
It follows that
$$\max_{\taxi_x\in T}(\val(\pass_i)-\overline{\ell}_{i,x}-\price_{\taxi_x}) \leq  0,$$
and thus
$$\val(\pass_i)\leq\min_{\taxi_x\in T}(\overline{\ell}_{i,x}+\price_{\taxi_x})=r_{\loc(\pass_i)}.$$

Consider some buyer $\pass_i$ that was allocated an item, $t_j$, in the social welfare maximizing allocation for $M(B,T)$.
It follows that $\max_{\taxi_x\in T}\auctionutility_{\pass_i}(\taxi_x)\geq 0$.
Thus,
$$\max_{\taxi_x\in T}(\val(\pass_i)-\overline{\ell}_{i,x}-\price_{\taxi_x}) \geq  0,$$
and
$$\val(\pass_i)\geq\min_{\taxi_x\in T}(\overline{\ell}_{i,x}+\price_{\taxi_x})=r_{\loc(\pass_i)}.$$
\end{proof}

\begin{lemma}\label{lemma:discretetaxibr}
For supply $s$, demand $d$,  surge prices $r$, and new supply $s'$ as defined above. A taxicab $t_j$ that serves passenger $\buyer(\taxi_j)$ is doing a best response. \end{lemma}

\begin{proof}

%We first consider the requirement regarding taxicabs,
%For $u,v$ such that $f^*(u,v)>0$ the utility $$\mu(u\mapsto v|r,s',d)\geq \mu(u\mapsto w|r,s',d)\qquad \forall w\in V.$$

Consider the following cases:
\begin{enumerate}\item
Item $\taxi_j$ is not allocated, {\sl i.e.}, $\buyer(\taxi_j)=\emptyset$. It follows that the Walrasian pricing for item $\taxi_j$ is zero: $\price_{\taxi_j}=0$. Now, for any $w\in V$ we have that
$$r_w=\min_{\taxi_x\in T}(\ell_{w,\loc(\taxi_x)}+\price_{\taxi_x})\leq\ell_{w,\loc(\taxi_j)}+\price_{\taxi_j}=\ell_{w,\loc(\taxi_j)},$$
hence, $r_w-\ell_{w,\loc(\taxi_j)}\leq 0$.
Ergo, not serving any passenger is a best response for $\taxi_j$.
\item Item $\taxi_j$ is allocated to some buyer $\pass_i$.  From Lemma \ref{lem:achievesmin} we know that $\overline{\ell}_{i,j}+\price_{\taxi_j}=\min_{\taxi_x\in T}(\overline{\ell}_{i,x}+\price_{\taxi_x})=r_{\loc(\pass_i)}$ and thus $\taxi_j$ gains a utility of $\price_{\taxi_j}$ from serving $\pass_i$.
If taxicab $\taxi_j$ could serve a passenger at location $w\in V$, it will gain a utility of
$$
r_w-\ell_{w,\loc(\taxi_j)}=
\min_{\taxi_x\in T}(\ell_{w,\loc(\taxi_x)}+\price_{\taxi_x})-\ell_{w,\loc(\taxi_j)}
\leq\ell_{w,\loc(\taxi_j)}+\price_{\taxi_j}-\ell_{w,\loc(\taxi_j)}=\price_{\taxi_j}.
$$
Implying that serving passenger $\pass_i$ is a best response for taxicab $\taxi_j$.
\end{enumerate}
\end{proof}

\begin{lemma}\label {lemma:discretetruth}
It is a dominant strategy for the passengers to reveal their true valuations given that surge prices are computed via the algorithm above. \end{lemma}

\begin{proof}
The utilities of the bidders for the minimal Walrasian prices in a unit demand market coincide with VCG payments \cite{Leonard83}. This implies that buyers truthfully reveal their valuations for the items. In our setting the utility for a passenger $\pass_i$  is exactly equal to the utility for the corresponding bidder $\pass_i$. Ergo, misreporting passenger valuation implies misreporting bidder valuations. As misreporting item valuations in the unit demand market setting cannot benefit buyers (and thus passengers) we conclude it is a dominant strategy for passengers to report true valuations.
\end{proof}

To summarize, our main result in this section, Theorem \ref{thm:finaldiscrete}, follows from Lemma \ref{lem:assignmentismax}, Lemma \ref{lemma:envyfree}, Lemma \ref{lemma:discretetaxibr}, and Lemma \ref{lemma:discretetruth}.

\begin{theorem} \label{thm:finaldiscrete}
For any Profile $P$ and supply $s$ there exist surge prices $r$, demand $d(r)$ and new supply $s'$ such that
\begin{itemize} \item Supply $s$, new supply $s'$, demand $d(r)$, and surge prices $r$ are in passenger-taxicab equilibrium.
\item $s'$ is social welfare maximizing with respect to supply $s$, profile $P$, and demand $d$.
\item The surge prices $r$ can be computed in polynomial time.
\item It is a dominant strategy for passengers to report their true valuations to the surge-price computation.
\end{itemize}
\end{theorem}

\section{Optimal Competitive Online Algorithms for Social Welfare}\label{sec:onlinealg}

In this section we give online algorithms that determine supply
(using surge prices) so as to maximize social welfare as given in
Definition \ref{def:socialwelfare}. {\sl I.e.}, striking a balance
between maximizing the quality of service {\sl vs.} the costs
associated with shifting resources about.

The results\footnote{These are randomized online algorithms.
Alternately, one could give deterministic online algorithms with the
same guarantees by using the passenger-taxicab equilibria and surge prices
derived from Theorem \ref{thm:anyst}, with the disadvantages that
the equilibria is no longer unique and that this requires some
additional technical assumptions.}  in this section can be obtained
by online algorithms that set the supply to be one of the following:
\begin{enumerate}
    \item Set supply at time $t$ equal demand at time $t$, {\sl i.e.}, set $s^t= d^t$.
    \item Set supply at time $t$ equal to the supply at time $t-1$, {\sl i.e.}, set $s^t=s^{t-1}$.
\end{enumerate}
It follows from Theorem \ref{thm:surge} that using appropriate surge
prices we can determine that $s^t=d^t$ as the unique
passenger-taxicab equilibrium. It is easy to leave the supply unchanged by
choosing $r^t_i=1$ for all $i$. It follows that the resulting
passenger-taxicab equilibrium has no positive flow from $i$ to $j\neq i$, as
 $\ell_{ij}\geq 1$ for all $j\neq i$ --- ergo
$s^t=s^{t-1}$.

Given a demand sequence $d$ we define $\rd$ as the inverse of the
maximum demand at any vertex and time, {\sl i.e.},
$1/\rd={\max_{i,t} d^t_i}$. Note that $\rd\leq k$ since at any
time $t$ there is a vertex $i$ such that $d^t_i\geq 1/k$.
Moreover, $\rd\geq 1$ since $d^t_i\leq 1$, for any time $t$ and
vertex $i$.

%In this section we will consider a few models for the demand in the
%online settings: (1) Restricted Demand per vertex --- We have a
%parameter $\rd\geq 1$ and the demand in a vertex never exceeds
%$\frac{1}{\rd}$. {\sl I.e.}, for all $i,t$, $d^t_i\leq \frac{1}{\rd}$.
%Note that
%$\rd=1$ only implies that $\|d^t\|_1=1$, and implicitly $\rd\leq k$
%since there is always $d^t_i\geq \frac{1}{k}$. (2)

Consider the following online algorithms:
\begin{description}
    \item [$\rand(p)$] --- With probability $p$ set surge prices such that supply equals demand at all vertices. {\sl I.e.},
    at time $t=1$ set $s^1=d^1$; for all $t>1$ with probability $p$ set $s^t=d^t$ and with probability $1-p$ set $s^t=s^{t-1}$.
    \item [$\stay$] --- Split the supply equally over all vertices.
    {\sl I.e.}, at time $t=1$ set $s^1=\langle \frac{1}{k},\frac{1}{k},\ldots,\frac{1}{k}\rangle$ and for all $t>1$ set $s^t=s^{t-1}$.
    \item [$\match$] --- Always set supply equal demand, {\sl i.e.},
    set $s^t=d^t$ for all $t\geq 1$ . Note that $\match$ and $\rand(1)$ are identical.
    \item [$\comp(p)$] --- Toss a fair coin, if heads run $\stay$ otherwise run $\rand(p)$.
    The expected social welfare of $\comp(p)$ satisfies $\expec[\comp(p)]=\expec[\stay]/2 + \expec[\rand(p)]/2$.
\end{description}

In different scenarios different algorithms are useful. We later
discuss how to switch between different online algorithms in
changing circumstances, varying over time.

%Algorithm $\comp$ would be the one used for the case of uniform
%distances and Restricted Demand per node, algorithm $\stay$ would be the one used for the case of arbitrary distances, and algorithm $\match$ would be used
%for the case of Restricted Drift.

Like many other online problems, we first show
that the optimal solution can be assumed to be
``lazy", never move supply about unnecessarily  (Section~\ref{sec:lazy}).
Section \ref{sec:uniform} gives our main technical result. In this
setting the cost of moving from one vertex to another always equals
$1$, i.e., $\ell_{ij}=1$ for $i\neq j$. In this scenario we show
that $\comp({\sqrt{1/{k}}})$ achieves [an optimal] $\Theta(1/\sqrt{k})$
fraction of the optimal social welfare. More generally, the
competitive ratio improves  as a function of the maximal demand in a
single vertex (a $1/\rho$ fraction of the total demand) --- in this
setting $\comp({\sqrt{{\rd}/{k}}})$ achieves [an optimal]
$\Theta(\sqrt{\rho/{k}})$  fraction of the optimal social welfare.
The positive result appears in Theorem \ref{thm:compfinal}, whereas
optimality follows from Lemma \ref{thm:upperuniform}.

In Section \ref{sec:extensions} we consider several other scenarios:
\begin{itemize}
\item
Clearly, even for completely arbitrary costs $\ell_{ij}$ (to move
supply from $i$ to $j$), algorithm $\stay$ is trivially $\rho/k$
competitive.  In Section \ref{sec:arbitrarymetric} we prove that
this cannot be improved. This shows that it is critical that
$\ell_{ij}=1$ to obtain a non-trivial bound, without other
assumptions on the input sequence.
\item
In Section \ref{sec:restricted} we consider inputs where the total
drift (average total variation distance between successive demand vectors) is
small. In such settings the $\match$ algorithm approaches the
optimal social welfare, for sufficiently small drift. Moreover,
essentially the same bounds are tight.
\end{itemize}

%Due to  lack of space, some of the proof of this section are
%deferred to Appendix \ref{sec:appmissing}.

\subinput{The Optimal Supply Sequence is Lazy}
\label{sec:lazy}

We define lazy sequences and show that without loss of generality
the optimal supply sequence is a lazy sequence. We have two types of
``non-lazy" actions: increasing supply in a
location with supply greater than demand (over supply), or
reducing supply in a location while creating over demand.  Both
actions can be avoided, without loss in social welfare. We start
by defining a lazy sequence.

\begin{definition}
A supply sequence is {\em lazy} if for any time $t$ and any $u,v\in V,u\neq v$ such
that $f^t(u,v)>0$ then both (1) $s^t_v\leq d^t_v$ and (2) $s_u^{t-1}>
d^t_u$.
\end{definition}

%We first show that the optimal sequence can be lazy. Namely, if
%there is a flow at time $t$ from $u$ to $v$, then the resulting
%supply at $v$ does not exceed the demand in $v$.

We show that for any supply sequence there exists a lazy
supply sequence whose social welfare is at least the social welfare of the
original sequence.

\begin{lemma}\label{lem:tolazy}
Fix a demand
sequence $d$. Given an arbitrary supply sequence $s$, there exists a lazy supply sequence
$\bar{s}$ such that $\sw(\bar{s})\geq \sw(s)$.
%%Formally,
%Fix any demand sequence. For any time $t$: if $f^t(u,v)>0$ then
%$s^t_v\leq d^t_v$.
\end{lemma}

%Due to  lack of space, the proof of this lemma is deferred to Appendix \ref{sec:appmissing}.

\begin{proof}
	For contradiction, assume there is a sequence $s$ for which for any
	lazy sequence $\bar{s}$ we have $\sw(s)>\sw(\bar{s})$. Note that
	essentially we are saying that there is  an optimal sequence $s$ for
	which no lazy sequence has the same social welfare. This implies
	that for any optimal sequence $s$ there is a time $t$ such that
	$f^t(u,v)>0$ and either (1) $s^t_v > d^t_v$ or (2) $s^{t-1}_u < d^t_u$.
	% optimal sequence there is a time $t$ such both $f^t(u,v)>0$ then $s^t_v > d^t_v$.
	Out of all the optimal sequences, consider the optimal sequence $s$
	with the largest such time $t$ and largest pair $(u,v)$ (given some
	full order on the pairs $V\times V$).
	
	We create a new flow $\bar{f}$ depending on the type of violation.
	Assume that we have $f^t(u,v)>0$ and $s^t_v > d^t_v$. At time $t$
	set $\bar{f}^t(u,v)= f^t(u,v)-\epsilon$ and $\bar{f}^t(u,u)=
	f^t(u,u)+\epsilon$, where $\epsilon=\min\{s^t_v - d^t_v,
	f^t(u,v)\}$. The rest of the flow remains unchanged, {\sl i.e.},
	$\bar{f}^t(u',v')= f^t(u',v')$ for $(u',v')\neq (u,v)$ or
	$(u',v')\neq (u,u)$.
	
	At time $t+1$ we adjust the flow to correspond to the original supply.
	Namely, for all $w\in V$ such that $f^{t+1}(v,w)>0$, we set
	$\bar{f}^{t+1}(v,w)=f^{t+1}(v,w)\frac{s^t_v-\epsilon}{s^t_v}$ and
	$\bar{f}^{t+1}(u,w)=f^{t+1}(u,w)+ f^{t+1}(v,w)
	\frac{\epsilon}{s^t_v}$, and all the remaining flows remain
	unchanged. It is straightforward to verify that $\bar{f}$ is a valid
	flow, and we set $s^{t+1}_v=\bar{s}^{t+1}_v=\sum_u \bar{f}^{t+1}(u,v)$.
	
	Note that the only influence on the social welfare are in times $t$
	and $t+1$. Comparing the movement cost of $\bar{s}$ to $s$, at time
	$t$ it decreased by $\epsilon$ and in time $t+1$ increased by at
	most $\epsilon$. The demand served in $\bar{s}$ and $s$ at time $t$
	and $t+1$ in unchanged (since the $\epsilon$ flow that was modified
	did not serve any demand in time $t$ and at time $t+1$ the supplies
	are identical). This implies that the social welfare of $\bar{s}$ is
	at least that of $s$. Therefore we have a contradiction to our
	selection of $t$ and $(u,v)$.
	
	The case that we have $f^t(u,v)>0$ and $s^{t-1}_u < d^t_u$ is similar
	and omitted.
\end{proof}

We derive the following immediate corollary.

\begin{corollary}
Without loss of generality the optimal supply sequence is lazy.
\end{corollary}

\subinput{Online Algorithms for Social Welfare Maximization when $\ell_{ij}=1$}
\label{sec:uniform}
%\begin{lemma}
%There exists some optimal supply sequence $OPT$ for which, for all
%$i,t$ :  $s^{t-1}_i<s^{t}_i \rightarrow d^t_i\geq s^t_i$.
%\end{lemma}
%
%[[LS - this might not be so easily. Not sure how I want to portray
%this - we have this variation to the sequence that we want to cause.
%However, it isn't easy to say exactly what change it is because it
%isn't flow supply from one specific city ot another...]]
%
%\begin{proof}
%    TODO - This is easily provable by taking some $OPT$ for which this does not hold and postponing the movement to the relevant node until such demand exists.
%\end{proof}
%
%Thus, from now on we will assume that $OPT$ satisfies the lemma.
%[[LS - we need to denote such an $OPT$ in some way to ease the usage of such an $OPT$ later on]]

We now analyse the lazy optimal supply sequence. We first introduce some notation.
Given an optimal lazy supply sequence $s$, define
$h^t_i=\min\{s^{t-1}_i,d^t_i\}$. Let $n\geq 0$ be an integer parameter, and define\footnote{For notational convenience we define $d^t_i=0$ and $s^t_i = s^1_i$ for all $t\leq 0$.} 
 $$z^t_i=\max\{0,
h^t_i-g^t_i\}, \mbox{\rm\ where\ } g^t_i=\max_{\tau\in[\max(1,t-n),t-1]}d^\tau_i.$$
Note that the definitions depend on $s$, but we use a fixed
optimal lazy sequence $s$. Note too that $n$ is yet undetermined.

\begin{lemma}
\label{lemma:opt}%
Fix a demand sequence $d$ and an optimal lazy supply sequence $s$ for $d$.
The resulting social welfare
\[
\opt=\sw(s,d)=\sum_{t,i} h_i^t \leq \sum_{t,i} z_i^t + \sum_{t,i}
g_i^t.
\]
\end{lemma}

\begin{proof}
Note that when $\ell_{ij}=1$ for all $i,j$ we get that $\ermv(s)=\sum_t
\frac{1}{2}\|s^t-s^{t-1}\|_1$. This means that for an optimal lazy
sequence we have
$$
\opt=\sw(s,d)=\ds(s,d)-\ermv(s)=\sum_t\sum_i
\min(s^{t}_i,d^t_i)-\sum_t\sum_{i:s^t_i\geq s^{t-1}_i}
\left(s^t_i-s^{t-1}_i\right).
%=\sum_t\sum_i \min(s^{t-1}_i,d^t_i).
$$
First consider $s^t_i > s^{t-1}_i$.
% then $\min(s^{t}_i,d^t_i)-(s^t_i-s^{t-1}_i)= \min(s^{t-1}_i,d^t_i)$.
Since the sequence is lazy and $s^t_i > s^{t-1}_i$ this implies that $s^t_i\leq
d^{t}_i$. Hence, $\min(s^{t}_i,d^t_i)=s^t_i$ and
$\min(s^{t-1}_i,d^t_i)=s^{t-1}_i$. It follows that the identity
$\min(s^{t}_i,d^t_i)-(s^t_i-s^{t-1}_i)= \min(s^{t-1}_i,d^t_i)$
holds.

Next consider $s^t_i< s^{t-1}_i$.
% then $\min(s^{t}_i,d^t_i)= \min(s^{t-1}_i,d^t_i)$.
Since the sequence is lazy and  $s^t_i< s^{t-1}_i$ implies that
$s^t_i\geq d^t_i$  and that
$\min(s^{t}_i,d^t_i)=d^t_i=\min(s^{t-1}_i,d^t_i)$. It follows yet again that the identity
$\min(s^{t}_i,d^t_i)= \min(s^{t-1}_i,d^t_i)$ holds.

Combining both identities we have
\[
\opt=\sw(s,d)=\sum_t\sum_i \min(s^{t-1}_i,d^t_i)= \sum_t\sum_i h^t_i,
\]
by the definition of $h_i^t$. Since, $h_i^t\leq z_i^t+g_i^t$ the
lemma follows.
\end{proof}

Our next goal is to bound the sum of $z_i^t$ and relate it to
the social welfare of the algorithm $\stay$. We first prove the following properties
of the optimal lazy supply sequence.

\begin{lemma}
\label{lemma:1}
Fix an optimal lazy sequence $s$ and a parameter $n\geq 1$. If for
some $i,t$ we have $s^{t-1}_i\geq\max_{\tau\in [t-n,t)} d^\tau_i$
then we have $\min_{\tau\in [t-n,t)} s^\tau_i\geq s^{t-1}_{i}$.
\end{lemma}

\begin{proof}
For contradiction assume there exists some maximal $\tau\in [t-n,t)$
such that $s^\tau_i< s^{t-1}_{i}$. Then, $\tau\neq t-1$ and thus
$\tau+1\in[t-n+1,t)$ which by the assumption of the lemma implies
that $s^{t-1}_i\geq d^{\tau+1}_i$. Also, because this is the maximal
such $\tau$ we have that $s^{\tau+1}_i\geq s^{t-1}_i$. Thus, we have
$s^{\tau}_i < s^{\tau+1}_i$ and $d^{\tau+1}_i < s^{\tau+1}_i$.
This contradicts the assumption that $s$ is an optimal lazy sequence,
since there is a flow to $i$ at time $\tau+1$ which strictly exceeds
the demand.
\end{proof}

We derive the following immediate corollary:

\begin{corollary}
\label{cor:2}
Fix an optimal lazy sequence $s$  and a parameter $n\geq 1$. If for
some $i,t$ we have $s^{t-1}_i\geq\max_{\tau\in [t-n,t)} d^\tau_i$
then for any $\tau\in [t-n+1,t)$ we have $ s^{\tau-1}_i\geq
s^{\tau}_{i}$.
\end{corollary}

\begin{proof}
From Lemma~\ref{lemma:1}, for any $\tau\in [t-n,t)$ we have that $s_i^\tau\geq s_i^{t-1}\geq
\max_{\tau'\in [t-n,t)} d^{\tau'}_i$. Therefore, $s_i^\tau\geq
\max_{\tau'\in [\tau-n',\tau)} d^{\tau'}_i$, where $n'=\tau-(t-n)>0$. Now
applying Lemma~\ref{lemma:1} again we obtain the corollary.
\end{proof}

\begin{lemma}
\label{lemma:sum-z}
Fix an optimal lazy sequence $s$ and a parameter $n\geq 1$. Then,
$\sum_i \sum_{\tau\in[t-n,t)} z_i^\tau \leq 1$.
\end{lemma}

\begin{proof}
	Clearly we care only about $z_i^\tau>0$. Fix a location $i$ and let
	$\tau_1, \ldots, \tau_m$ be all the times $\tau\in [t-n,t)$ for
	which $z_i^\tau>0$. Clearly, $ \sum_{\tau\in[t-n,t)} z_i^\tau =
	\sum_{j=1}^m z_i^{\tau_j}$.
	
	First, if $s^{t-1}_i\leq \max_{\tau\in [t-n,t)} d^\tau_i=g_i^t$,
	since $h_i^t\leq s^{t-1}_i$ then $z_i^t=0$. Therefore, at any time
	$\tau_j$ we have $s^{\tau_j-1}_i> \max_{\hat{\tau}\in [t-n,t)}
	d^{\hat{\tau}}_i$, which implies that we can apply
	Corollary~\ref{cor:2} at the times $\tau_j$.
	
	We claim that $s_i^{\tau_j-1}> d_i^{\tau_j}$ for $1\leq j\leq m-1$.
	For contradiction assume that $s_i^{\tau_j-1}\leq d_i^{\tau_j}$. We
	have
	\[
	h_i^{\tau_m} \leq s^{\tau_m-1}_i  \leq s^{\tau_j-1}_i\leq
	d_i^{\tau_j}\leq g^{\tau_m}_i,
	\]
	where the first inequality is from the definition of $h$, the second
	follows from Corollary~\ref{cor:2}, the third from our assumption,
	and the fourth from the definition of $g$. This implies that
	$z_i^{\tau_m}=\max\{0,h_i^{\tau_m}-g_i^{\tau_m}\}=0$. In
	contradiction to our construction that $z_i^{\tau_m}>0$. Therefore,
	$s_i^{\tau_j-1}> d_i^{\tau_j}$, which implies that
	$h_i^{\tau_j}=d_i^{\tau_j}$.\footnote{This applies only to $j\leq
		m-1$ since $g^{\tau_m}_i$ does not include $d_i^{\tau_m}$ but does
		include all previous $d_i^{\tau_j}$.}
	
	Since $z_i^{\tau_j}>0$, we have that
	$z_i^{\tau_j}=h_i^{\tau_j}-g_i^{\tau_j}$. We showed that
	$h_i^{\tau_j}=d_i^{\tau_j}$ and $g_i^{\tau_j}\geq d_i^{\tau_{j-1}}$,
	hence, $z_i^{\tau_j}\leq d_i^{\tau_{j}}-d_i^{\tau_{j-1}}$, for
	$2\leq j\leq m-1$.
	
	Summing over all $\tau_j$ we have
	\begin{align*}
	\sum_{\hat{\tau}\in[t-n,t)} z_i^{\hat{\tau}} &= \sum_{j=1}^m
	z_i^{\tau_j} \\
	&= z_i^{\tau_m} + z_i^{\tau_1}+\sum_{j=2}^{m-1}
	z_i^{\tau_j} \\
	&\leq z_i^{\tau_m} +z_i^{\tau_1}+ \sum_{j=2}^{m-1}
	d_i^{\tau_{j}}-d_i^{\tau_{j-1}}\\
	&\leq z_i^{\tau_m} +z_i^{\tau_1}+ d^{\tau_{m-1}}_i -d^{\tau_1}_i\\
	&\leq h^{\tau_m}_i -(g^{\tau_m}_i -d^{\tau_{m-1}}_i) + (h^{\tau_1}_i
	-g^{\tau_1}_i -d^{\tau_{1}}_i)  \\
	&\leq h^{\tau_m}_i
	\end{align*}
	For the last inequality note that $g^{\tau_m}_i \geq
	d^{\tau_{m-1}}_i$ and that $h^{\tau_1}_i \leq d^{\tau_{1}}_i$.
	
	Summing over all locations $i$ we have
	\[
	\sum_i \sum_{\hat{\tau}\in[t-n,t)} z_i^{\hat{\tau}}\leq \sum_i
	h^{\tau_{m_i}}_i \leq \sum_i s^{\tau_{m_i}}_i \leq \sum_i s^{t-n}_i
	=1
	\]
	where the last inequality uses again Corollary~\ref{cor:2}.
\end{proof}

We now analyze $\stay$ for arbitrary relocation costs $\ell_{ij}$.

\begin{lemma}
\label{lemma:stay}%
 At all times $t$, the demand served by $\stay$ is at least 
$\rd/k$ of the total demand.
\end{lemma}

\begin{proof}
Recall that $\ds\left(s^t,d^t\right)=\sum_i
\min\left(s^t_i,d^t_i\right)=\sum_i
\min\left(\frac{1}{k},d^t_i\right)$. Denote $S=\left\{i|s^t_i\geq
\frac{1}{k}\right\}$.
If we have $|S|\geq \rd$ then $\ds\left(s^t,d^t\right)\geq
\frac{1}{k}\cdot |S|\geq \frac{\rd}{k}$.
%
    %TODO - fix this up - The beginning is the worst case scenario
    %TODO - at some point we also use the fact the r<=k (to say that r-k is negative) and we use j<=r
Otherwise, since $\frac{1}{\rd}\geq \frac{1}{k}$ the total demand
not in $S$ is at least $1-\frac{|S|}{\rd}$ and it is completely
served by $\stay$. Therefore,
\[\ds\left(s^t,d^t\right)\geq
|S|\cdot\frac{1}{k}+1-\frac{|S|}{\rd}=\frac{k\rd+|S|\rd-|S|k}{k\rd}=\frac{k\rd-|S|(k-\rd)}{k\rd}\geq\frac{k\rd+\rd^2-k\rd}{k\rd}=\frac{\rd}{k}.
\]
\end{proof}

Now we analyze $\rand(p)$ and relate it to $g_i^t$.
\begin{lemma}
\label{lemma:rand}%
Let $\widehat{s}_i^t$ be the random variable representing the supply
of $\rand(p)$ time $t$ in vertex $i$. Then, $\expec
[\widehat{s}_i^t]\geq g_i^t p (1-p)^n$. In addition, the expected
social welfare of $\rand(p)$ is at least $p(1-p)^n \sum_{i,t}
g_i^t$.
\end{lemma}

\begin{proof}
Let $\tau=\arg\max_{\hat{\tau}\in[t-n,t)} d_i^{\hat{\tau}}$, {\sl
i.e.}, $d^\tau_i=g^t_i$. We lower bound the expectation of
$\widehat{s}_i^t$ by the probability that $\rand(p)$ sets
$s^\tau=d^\tau$ and keeps the supply until time $t$, {\sl i.e.},
$s^t=s^\tau$. The probability that we have $s^\tau=d^\tau$ is
at least $p$. The probability that $s^t=s^\tau$ is at least
$(1-p)^n$. Therefore, $\expec [\widehat{s}_i^t]\geq g_i^t p
(1-p)^n$, which implies that the expected social welfare of
$\rand(p)$ is at least $p(1-p)^n \sum_{i,t} g_i^t$.
\end{proof}

\begin{theorem}
The algorithm $\comp({\sqrt{{\rd}/{k}}}) =\frac{1}{2}\stay +
\frac{1}{2}\rand({\sqrt{{\rd}/{k}}})$ is
$(\frac{1}{2e}\sqrt{\frac{\rd}{k}})$-competitive.
\end{theorem}\label{thm:compfinal}

\begin{proof}
By Lemma~\ref{lemma:opt} we have that
%
%Recall that
%\[
%OPT=\sw(s,d)=DS(s,d)-\ermv(s)=\sum_t\sum_i
%\min(s^{t}_i,d^t_i)-\sum_{i:s^t_i\geq s^{t-1}_i}
%\left|s^t_i-s^{t-1}_i\right|=\sum_t\sum_i \min(s^{t-1}_i,d^t_i).
%\]
%This implies that
$OPT= \sum_{t,i} h_i^t \leq \sum_{t,i} z_i^t +
g_i^t$. We bound separately $\sum_{t,i} z_i^t$ and $\sum_{t,i}
g_i^t$.

By Lemma~\ref{lemma:sum-z} we can partition the time to
$\frac{T}{n}$ blocks of size $n$ each, and in each the sum is at
most $1$, therefore $\sum_{t,i} z_i^t\leq \frac{T}{n}$. On the other
hand, $\stay$ guarantees a social welfare of at least $\rd\cdot
\frac{T}{k}$.

%For each time $t$ we show that $s_i^t$, the supply of $\rand(p)$ has
%$E[s_i^t]\geq g_i^t p (1-p)^n$. Let
%$\tau=\arg\max_{\hat{\tau}\in[t-n,t)} d_i^{\hat{\tau}}$. We can
%lower bound the expectation of $s_i^t$ by the probability that \rand
%will set $s^\tau=d^\tau$ and keep the supply until time $t$, i.e.,
%$s^t=s^\tau$. The probability that we have $s^\tau=d^\tau$ is
%exactly $p$. The probability that $s^t=s^\tau$ is at least
%$(1-p)^n$. Therefore, $E[s_i^t]\geq g_i^t p (1-p)^n$, which implies
%that the expected social welfare of $\rand(p)$ is at least $p(1-p)^n
%\sum_{i,t} g_i^t$.

We have that,
\[
OPT \leq \frac{T}{n} + \sum_{i,t} g_i^t.
\]
Using Lemma~\ref{lemma:stay} and Lemma~\ref{lemma:rand}, we have
\[
\frac{1}{2}\stay + \frac{1}{2}\rand(p) \geq \frac{\rd}{2k}T +
\frac{1}{2}p(1-p)^n \sum_{i,t} g_i^t
\]
For $p=\sqrt{\frac{\rd}{k}}$ and $n=\frac{1}{p}$ we bound the competitive ratio as follows:
\[
\frac{\rd \frac{T}{2k} + \frac{1}{2}p(1-p)^n \sum_{i,t} g_i^t}{\frac{T}{n} + \sum_{i,t}
g_i^t} =\frac{\frac{1}{2}\sqrt{\frac{\rd}{k}} T\sqrt{\frac{\rd}{k}} + \frac{1}{2e}\sqrt{\frac{\rd}{k}}
\sum_{i,t} g_i^t}{T\sqrt{\frac{\rd}{k}} + \sum_{i,t} g_i^t}\geq
\frac{1}{2e}\sqrt{\frac{\rd}{k}}.
\]
%
%Assume some demand sequence $Q=\langle d^1,d^2,\ldots, d^T\rangle$.
%Also, some optimal offline algorithm $OPT$ which has supply sequence
%$\hat{S}=\langle \hat{s}^1,\hat{s}^2,\ldots,\hat{s}^T\rangle$.
%Denote $w^{t}_i=\max_{t'\in[t-\sqrt{\frac{k}{r}},t-1]}
%\left(d^{t'}_i\right)$ and $z^t_i=
%\max\left(\hat{s}^{t-1}_i-w^t_i,0\right)$. Thus, $w^t_i+z^t_i\geq
%\hat{s}^{t-1}_i\geq \min\left(\hat{s}^{t-1}_i,d^t_i\right)$ and
%$sum_{t,i} w^t_i+z^t_i\geq sum_{t,i}
%\min\left(s^{t-1}_i,d^t_i\right) = \sw\left(Q,\hat{S}\right)$
%
%Let us prove then that $\sum_{i,t} z^t_i \leq
%\sqrt{\frac{r}{k}}\cdot T$,
%
%Let us prove then that $\mathop{{}\mathbb{E}} s^t_i \geq
%O\left(\sqrt{\frac{r}{k}}\right)\cdot w^t_i$, denote $\tau$ such
%that $d^\tau_i=w^t_i$. Then, since at each time $t$
%$\rand({\sqrt{\frac{r}{k}}})$ will set $s^t=d^t$ with probability
%$\sqrt{\frac{r}{k}}$ we have
\end{proof}

\subinput{Social Welfare Maximization when $\ell_{ij}=1$: Impossibility Results}

%In this section we will assume $ \ell_{ij} = 1 $ for every $ i,j$,
%demand at each node restricted to $\frac{1}{\rd}$ ({\sl i.e.} $d^t_i\leq\rd$)
%and an unrestricted $ \delta $.

We show that no online algorithm can hope to achieve a competitive
ratio better (greater) than $O\left(\sqrt{\frac{\rd}{k}}\right)$.
Recall, that Section \ref{sec:uniform} describes an online
algorithm, $\comp({\sqrt{{\rd}/{k}}})$, that achieves this bound on the competitive
ratio. Ergo, $\comp({\sqrt{{\rd}/{k}}})$ achieves the optimal competitive ratio, up to a
constant factor.

\begin{theorem}\label{thm:upperuniform}
    Fix the metric $\ell_{ij}=1$.
        No online algorithm can achieve a competitive ratio better (greater) than $O\left(\sqrt{\frac{\rd}{k}}\right)$.
\end{theorem}

\begin{proof}
	We first describe the proof for $\rd=1$ and then extend it to arbitrary $\rd$.
	
	Consider the following stochastic demand sequence. At time $t$ we
	select at random a vertex $c^t\in V$, and assign all the demand to it,
	{\sl i.e.}, $d^t_{c^t}=1$ and $d^t_i=0$ for $i\neq c^t$. Clearly any
	online algorithm has an expected social welfare of $T/k$.
	
	Essentially, for the optimal offline we use the birthday
	paradox to show that its social welfare is $\Theta(T/\sqrt{k})$.
	Consider the following offline strategy. Partition the time to
	intervals of size of $2\sqrt{k}$. We show that in any such
	interval the offline can increase social welfare by at least $1$ with constant
	probability.
	
	Fix such a time interval. We claim that with constant probability
	some vertex appears twice in the interval. If in the first
	$\sqrt{k}$ times there is a vertex $i$ that appears twice, we are
	done.
	%the offline will move at the start of the interval to vertex $i$ and will gain at least $1$.
	Otherwise, we have $\sqrt{k}$ distinct vertices. The probability that
	we resample one of those vertices in the next $\sqrt{k}$ time steps is
	at least $1/e$. Now, if vertex $i$ appears twice in the interval then
	the offline algorithm can move at the start of the interval to vertex $i$ and
	increase social welfare by at least $1$. This implies that the expected social
	welfare of this offline strategy is $\Theta(T/\sqrt{k})$, which
	lower bounds the expected social welfare of the optimal offline
	strategy.
	
	Since the online algorithm has expected social welfare of $T/k$ and the
	optimal offline algorithm has expected social welfare of $\Theta(T/\sqrt{k})$,
	the competitive ratio, for $\rd=1$, is bounded by $O(\sqrt{1/k})$.
	
	%%%%%%%%%%%%%%%%%%%%%%%%%%%
	
	We now sketch how the proof extends to a general $\rd\geq 1$. In
	this case we partition the $k$ vertices into
	$N=\floor{k/\ceil{\rd}}$ disjoint subsets, each of size
	$M=\ceil{\rd}$. (Note, that $N\cdot M\leq k$.) The $N$ subsets
	replace the vertices $V$ and each time we select a subset, we give a
	uniform demand over the subset. (note that the demand per vertex is $1/M\leq 1/\rd$.)
	
	As before, any online algorithm has expected social welfare of
	$\Theta(T/N)=\Theta(T\rd/k)$. Similar to before, there is an offline
	strategy that guarantees an expected social welfare of
	$\Theta(T\sqrt{\rd/k})$. This implies that the competitive ratio is
	at most $\Theta(\sqrt{\rd/k})$.
\end{proof}

\subinput{Extensions}\label{sec:extensions}

In (Section \ref{sec:arbitrarymetric} we show that the assumption that $\ell_{ij}=1$ was critical to achieve the non-trivial competitive ratio of Section \ref{sec:uniform} unless $\rho$ (the fraction of demand at any single vertex) was sufficiently small.  We also consider restricting the demand sequences by bounding the average variability in demand. In Section \ref{sec:restricted} we show that the online algorithm that greedily matches supply and demand works well, the average drift is sufficiently small.

\subsubinput{Arbitrary Metric Spaces} \label{sec:arbitrarymetric}

We can apply the online algorithm $\stay$ and guarantee a competitive
ratio of $\rd/k$ as shown in Lemma~\ref{lemma:stay}. The following
theorem establishes an impossibility result when the costs are
different than $1$ (even if they are still identical).

\begin{theorem}\label{thm:arbmetric}
    Fix some $1>\epsilon>0$, and consider costs $\ell_{ij}=1+\epsilon$ for $i\neq j$.
No online algorithm has a competitive ratio better (greater) than
$\frac{(1+\epsilon)^2}{\epsilon}\cdot\frac{1}{k}$ for this metric.
\end{theorem}

\begin{proof}
	The idea is the following: we generate a demand sequence that at every time step demand is concentrated in a single vertex.  We generate a
	random sequence of vertices, such that no two successive positions are identical. We then duplicate every position for a random duration. The duration, the number of successive demands at that position, is geometrically distributed. We set the parameters such that
	no online algorithm can benefit by switching between vertices.
	On the other hand, given a sufficiently long duration of repeated demands for the same vertex, the optimal
	schedule switches to this vertex.
	
	We now describe the stochastic demand sequence generation. We first
	generate a sequence of locations $c$. We set $c_1= i\in V$ uniformly
	at random. For $c_\tau$ we set $c_\tau=j$ where $j\in V\setminus
	\{c_{\tau-1}\}$ uniformly. In addition we generate a sequence of
	duration $b$ distributed geometrically with parameter
	$p=\frac{1}{1+\epsilon}$. Namely, $b_\tau=j$ with probability
	$p^{j-1}p$, for $j\geq 1$. We are now ready to generate the demand
	sequence $d$. For each $c_\tau=i$ we associate a unit vector $e_i$
	which has $e_{i,i}=1$ and $e_{i,j}=0$ for $j\neq i$. We duplicate
	$e_{c_\tau}$ exactly $b_\tau$ times. We truncate the sequence at
	time $T$, and this is the demand sequence $d$.
	
	%
	%Generate the following demand sequence $Q$:
	%\begin{itemize}
	%    \item For time $t=1$ choose some $i$ uniformly at random from $\left[1,k\right]$, denote $c_1=i$. Set $d^1_i=1$ and $d^1_j=0$ for all $j\neq i$.
	%    \item For time $t>1$, choose $c_t$ as follows:
	%    $$c_{t} = \left\{\begin{array}{cl}
	%    c_{t-1} & \mbox{\rm With probability $\frac{\epsilon}{1+\epsilon}$} \\
	%    $j$ & \mbox{\rm With probability $\frac{1}{(1+\epsilon)(k-1)}$, $j\in\{1,\ldots,k\}
	%        \setminus \{c_t\}$}
	%    \end{array}\right.$$
	%   \item For all $t$ set $d^t_{c_t}=1$ and for all $j\neq c_t$ set $d^t_j=0$.
	%\end{itemize}
	
	First consider an arbitrary online algorithm. We claim that it does
	not gain (in expectation) any social welfare by moving supply, and
	hence it's expected social welfare is $T/k$. The argument is that the
	cost of moving $\delta$ supply to a new location is $(1+\epsilon)\delta$. On
	the other hand, the expected duration in the new location is only
	$1+\epsilon$, so in expectation there is no benefit. For an online
	algorithm that does not move any supply the expected social welfare
	is $T/k$.

	%At time $t$ an arbitrary online algorithm may do one of the
	%following:
	%%TODO - fix the eplan
	%\begin{enumerate}
	%\item
	%Move supply to $c_t$ from one (or more) arbitrary nodes $j$ with
	%$s^{t-1}_j>0$. Any supply added to $c_t$ will serve $(1+\epsilon)$
	%in expectation until the random sequence switches the demand
	%elsewhere. As the earthmover cost is equal to $(1+\epsilon)$ times
	%the supply moved, the net effect of such a choice is zero.
	%Subsequent to the next change in $c_t$ the situation is exactly as
	%before. Thus, without loss of generality no online algorithm will
	%move supply to node $c_t$.
	%\item
	%Any other change in the supply (not to $c_t$) will imply a negative
	%net change in the expected social welfare.
	%\end{enumerate}
	
	We now analyze the social welfare attained by an optimal offline
	algorithm. The main benefit of an offline algorithm is that it has
	access to the realized $b=\tau$. It is simple to see that if
	$b_\tau\geq 2$ then the offline algorithm has a benefit of $b_\tau -
	(1+\epsilon)>0$.

	%
	%Consider a sequence $d^t, d^{t+1}, \ldots, d^{t+i-1}$ where $c_t \neq
	%c_{t-1}$ and $c_j=c_{j'}$ for all $j,j' \in t,\ldots,t+i-1$. Any
	%algorithm that moves all supply to $c_t$ at time $t$ will lose
	%earthmover distance $(1+\epsilon)$ and serve $i$ demand.
	%
	%We call this a sequence of length $i$ The probability for a
	%sequence to be of length $i$ for some $i>1$ is
	%$(\frac{\epsilon}{\epsilon + 1})^{i-1} \cdot \frac{1}{1+\epsilon}$,
	%for each such sequence social welfare will increase by $i$ and decrease by $1+\epsilon$, netting a total of $i-1-\epsilon$ increase in social welfare.
	%Total social welfare will thus be:
	
	\begin{eqnarray}
	\expec[b_\tau-(1+\epsilon)|b_\tau\geq
	2]\Pr[b_\tau\geq 2] =\sum_{i=2}^{\infty} (\frac{\epsilon}{\epsilon +
		1})^{i-1} \cdot \frac{i-1-\epsilon}{1+\epsilon}=
	%\sum_{i=2}^{\infty} (\frac{\epsilon}{\epsilon + 1})^{i-1} \cdot \frac{-\epsilon}{1+\epsilon} + \sum_{i=2}^{\infty} (\frac{\epsilon}{\epsilon + 1})^{i-1} \cdot \frac{i-1}{1+\epsilon} \nonumber\\ &=&\epsilon\cdot\frac{-\epsilon}{1+\epsilon}+\frac{\frac{\epsilon}{\epsilon + 1}}{(1-\frac{\epsilon}{\epsilon + 1})^2}\cdot\frac{1}{1+\epsilon}\label{eq:2nd}\\
	%&=& \frac{-\epsilon^2}{1+\epsilon}+\frac{\epsilon\cdot(1+\epsilon)}{1+\epsilon} \nonumber\\
	\frac{\epsilon}{1+\epsilon}.\nonumber
	\end{eqnarray}
	
	We now would like to sum over $\tau$ however the numbers summands in
	the sum is a random variable. Since we have a random sums of random
	variables we need to use Wald's identity. Since the expected number
	of summands is $\frac{T}{1+\epsilon}$ and the expectation of each is
	$\frac{\epsilon}{1+\epsilon}$ we have that the optimal offline
	algorithm has an expected social welfare of at least
	$\frac{\epsilon}{(1+\epsilon)^2} T$.
	
	%
	%
	%Equation \eqref{eq:2nd} follows from simple geometric series calculations. This means an expected social welfare increase
	%of $\frac{\epsilon}{1+\epsilon}$ per sequence and a sequence is of
	%length $1+\epsilon$ in expectancy. Thus our total social welfare
	%will be $\frac{\epsilon}{(1+\epsilon)^2}\cdot T$ in expectation. For $\epsilon>1$
	%we can expect better returns (We will omit the places where the
	%earning will be negative ({\sl i.e.} series of length smaller than $1+\epsilon$)).
	%
	%Thus, the offline algorithm will earn $\frac{\epsilon}{(1+\epsilon)^2}\cdot T$ in expectation and any online only $\frac{1}{k}\cdot T$ in expectation giving a ratio of $\frac{(1+\epsilon)^2}{\epsilon}\cdot\frac{1}{k}$.
	%
	This implies that no algorithm has a competitive ratio better than
	$\frac{(1+\epsilon)^2}{\epsilon}\frac{1}{k}$.
\end{proof}

\subsubinput{Restricted Drift} \label{sec:restricted}

For any demand sequence $d$ let $\delta\leq 1$ be the average drift, {\sl i.e.},
 $\sum_{t} \|d^t - d^{t-1}\|_{tv}= (1/2)\sum_{t} \|d^t - d^{t-1}\|_1 = \delta T$.

\begin{theorem}\label{thm:restriced drift}
For the case where costs $\ell_{ij} =1$ for all $i\neq j$, setting demand and supply equal (the $\match$ algorithm) gives social
welfare of $(1-\delta)T$, and is $(1-\delta)$-competitive.

For arbitrary $\ell_{ij}$, where $\ell_{ij}\leq \ell_{\max}$, the
$\match$ algorithm has social welfare of at least
$(1-\delta\ell_{\max})T$, and is
($1-\delta\ell_{\max}$)-competitive.
\end{theorem}

\begin{proof}
	Since for $\ell_{ij}=1$ the earthmover distance metric coincides
	with the total variation metric, we have that at time $t$ the social
	welfare of $\match$ is
	$1-\|d^t-s^{t-1}\|_{tv}=1-\|d^t-d^{t-1}\|_{tv}$ since $\match$ sets
	$s^{t-1}=d^{t-1}$. Summing over all time steps we get that the
	social welfare of $\match$ is $T-\delta T$.
	Since the social welfare of $\opt$ is at most $T$ we have that
	$\match$ is $(1-\delta)$-competitive.
	
	For a general metric, note that $\ermv(d^t,d^{t-1})\leq \ell_{\max}
	\|d^t-d^{t-1}\|_{tv}$. This implies that the social welfare of
	$\match$ is at least $(1-\ell_{\max}\delta) T$, and hence it is
	$(1-\ell_{\max}\delta)$-competitive.
\end{proof}

%\begin{proof}
%Since for $\ell_{ij}=1$ the earthmover distance metric coincides
%with the total variation metric, we have that at time $t$ the social
%welfare of $\match$ is
%$1-\|d^t-s^{t-1}\|_{tv}=1-\|d^t-d^{t-1}\|_{tv}$ since $\match$ sets
%$s^{t-1}=d^{t-1}$. Summing over all time steps we get that the
%social welfare of $\match$ is $T-\delta T$.
%%
%Since the social welfare of $OPT$ is at most $T$ we have that
%$\match$ is $(1-\delta)$-competitive.
%
%For a general metric, note that $\ermv(d^t,d^{t-1})\leq \ell_{\max}
%\|d^t-d^{t-1}\|_{tv}$. This implies that the social welfare of
%$\match$ is at least $(1-\ell_{\max}\delta) T$, and hence it is
%$(1-\ell_{\max}\delta)$-competitive.
%\end{proof}

\begin{theorem}\label{thm:upper bound ellij1 drift}
For the metric $\ell_{ij}=1$, no online algorithm has a competitive ratio
better (greater) than $1-\delta/4$.
\end{theorem}

\begin{proof}
	Consider the following demand sequence. The demand sequence uses
	only the first two locations, {\sl i.e.}, for all locations $i\neq1,2$ and
	times $t$ we have $d^t_i=0$. For each time $t$ we select the demand
	randomly from the following distribution.
	$$d^t = \left\{\begin{array}{cl}
	d^t_1=1,d^t_2=0 & \mbox{\rm With probability $\frac{1}{2}$} \\
	d^t_1=1-2\delta,d^t_2=2\delta & \mbox{\rm With probability $\frac{1}{2}$}
	\end{array}\right..
	$$
	The generated sequence has an expected drift of $\delta T$. Any
	online algorithm $ALG$ has, in expectation, social welfare of
	$(1-\delta)T$. The main point is that $\opt$ has a strictly
	better expected social welfare.
	
	Consider the online algorithm $\match$ as a starting point. Partition
	the time to $T/2$ pairs of time slots, $[2m-1,2m]$. Consider the
	event that $d^{2m-2}=d^{2m}\neq d^{2m-1}$. This event occurs with
	probability $1/4$. In such an event we can modify $\match$ and at
	time $2m-1$ set $s^{2m-1}=d^{m}$. (This requires knowing the future,
	but we are interested in $\opt$ so it is fine.) Such a modification
	increases the social welfare by $2\delta$ (lowering the serviced
	demand by $2\delta$ and lowering the movement costs by $4\delta$).
	Therefore, the expected social welfare is improved by
	$(1/4)(2\delta)(T/2)$. This implies that the expected social welfare
	of $\opt$ is at least $(1-(3/4)\delta)T$.
	
	%Also, $OPT$ will earn $1-\frac{\delta}{3}$. %TODO: prove this (easy probability proof)

	This means that no algorithm is more than
	$\frac{1-\delta}{1-(3/4)\delta}$-competitive. This implies that no
	online algorithm can have a competitive ratio better than
	$(1-\delta/4)T$.
	
\end{proof}

\section{Discussion}\label{sec:disc}
Social welfare in our setting depends on the taxicabs and their locations (the supply $s$),  passengers, their locations and values (the profile $P$), and distances between taxicabs and passengers.
In this \paper{} we introduce passenger-taxicab equilibria, prove their existence and give poly time algorithms for computing surge prices so as to maximize social welfare.

We have shown that although time series are a critical part of the social welfare gains of any taxicab provider, no algorithm can hope to achieve significant worst-case ratios. Thus, in the future different relaxations to the problem might be considered in order to allow for more adaptive algorithms.

%Our abstraction maximizes social welfare for a single set of such parameters. Clearly, it would be more natural to consider fluctuations over time. Note that the greedy approach of setting surge prices to maximize social welfare per time step does not imply maximizing the social welfare overall.

%When considering a sequence of fluctuations, taxicabs care about the future and not only about the current requests.
 %Why travel a great distance if there is only one passenger to serve?
%Clearly, it is much better to serve a passenger at a location where there will be other passengers to serve over  a long period of time.

%Thus, to maximize social welfare over time, it is natural to use competitive analysis of online algorithms, so as to determine a sequence of surge prices. The goal is to maximize social welfare over time, and compare the social welfare over time of the online sequence of surge prices to the optimal offline sequence of surge prices. We view this as a highly interesting research direction.

When computing the surge prices above, we have implicitly assumed that taxicab locations are known (e.g., via GPS).
Contrawise, passengers have no incentive to misreport their location (trivially) and valuation (as proved above).
An interesting variation on our models
would be to consider taxicabs declaring their own distances to passengers. Those would not be physical distances but rather a personalized cost for service at a given location. 

If such personalized costs are verifiable, and social welfare is redefined as the sum of passenger values served minus the personalized service costs, then the surge prices computed in this \paper{} maximize this new social welfare. This allows for more robust pricing mechanisms which allow us to incorporate issues such as ``start up costs" which are a bonus for drivers to get out of bed.

Taxicab personalized costs are private to the taxicab. Thus, any surge price computation would have to contend with private values of the taxicabs as well as private values for the passengers. It is easy to see that without Bayesian assumptions on the private values, little can be done. Just consider a passenger and a taxicab at the same location, they need to agree upon a price. In the Bayesian setting this is called the bilateral trading problem and there is a rich literature on the topic.

%\input{figures}

%\chapter{Appendix}\label{sec:appendix}

%\input{appendix}

\bibliographystyle{abbrv}
\bibliography{surge}
%\newpage{}

%\includepdf[pages=-]{hebrew_part}

\end{document}